\documentclass[12pt,a4paper]{article}
\usepackage{amsfonts,amssymb,amsmath,amsthm,cite}
\usepackage{graphicx}
\usepackage{times}
\evensidemargin=\oddsidemargin
\DeclareMathOperator{\Id}{Id}              
\newtheorem{assumption}{Assumption}[section]
\newtheorem{theorem}[assumption]{Theorem}
\newtheorem{corollary}[assumption]{Corollary}
\newtheorem{example}[assumption]{Example}
\newtheorem{lemma}[assumption]{Lemma}
\newtheorem{definition}[assumption]{Definition}
\newtheorem{prop}[assumption]{Proposition}
\newtheorem{remark}[assumption]{Remark}

\renewcommand{\a}{\alpha}                    
\renewcommand{\b}{\beta}                    
\newcommand{\B}{\mathcal{B}}              
\newcommand{\CC}{\mathcal{C}}              
\newcommand{\CA}{\mathcal{A}}
\newcommand{\CB}{\mathcal{B}}
\newcommand{\CE}{\mathcal{E}}
\newcommand{\CH}{\mathcal{H}}
\newcommand{\CP}{\mathcal{P}}
\newcommand{\CF}{\mathcal{F}}
\newcommand{\CS}{\mathcal{S}}          
\newcommand{\acts}{\triangleright}

\newcommand{\<}{\langle}       
     
\newcommand{\F}{\mathcal{F}}          

\newcommand{\cop}{\triangle}                 
\newcommand{\ts}{\otimes}                 

\newcommand{\bq}{\begin{eqnarray}}
\newcommand{\eq}{\end{eqnarray}}
\newcommand{\be}{\begin{equation}}
\newcommand{\ee}{\end{equation}}
\newcommand{\enq}{\nonumber\end{eqnarray}}

\newcommand{\ks}{*_\kappa}
\def\<#1,#2>{\langle#1\,,\,#2\rangle}      
\newcommand{\twobytwo}[4]{\begin{pmatrix}#1 & #2 \\ #3 & #4
                            \end{pmatrix}} 

\title{Star product realizations of $\kappa$-Minkowski space}
\date{}

\author{
B.Durhuus\thanks{durhuus@math.ku.dk} \\ Department of Mathematical Sciences,\\ 
University of Copenhagen,\\ Universitetsparken 5, DK-2100 Copenhagen,
Denmark
\and
A.Sitarz\thanks{sitarz@if.uj.edu.pl} 
\\ Institute of Physics,\\ Jagiellonian University,\\
Reymonta 4, 30-059 Krak\'ow, Poland
}
\begin{document}
\maketitle

\begin{abstract}
We define a family of star products and involutions associated with $\kappa$-Minkowski space. 
Applying corresponding quantization maps we show that these star products  
restricted to a certain space of Schwartz functions have isomorphic
Banach algebra completions. For two particular star products it is demonstrated 
that they can be extended  to a class of polynomially bounded smooth functions allowing 
a realization of the full Hopf algebra structure on $\kappa$-Minkowski space. Furthermore, 
we give an explicit realization of the action of the $\kappa$-Poincar\'e algebra 
as an involutive Hopf algebra on this representation of $\kappa$-Minkowski space 
and initiate a study of its properties.  
\end{abstract}

\section{Introduction}\label{Intr}
The $\kappa$-deformation of Minkowski space was originally proposed in
\cite{MajRue} as a Hopf algebra whose underlying algebra is
the enveloping algebra of the Lie algebra with generators 
$x_0,\dots, x_{d-1}$ fulfilling 

\begin{align}
\label{kappa_Mink}
[x_0,\,x_i]= \tfrac{i}{\kappa} \,x_i, \quad [x_i,\,x_j]=0, \quad i,j=1,\dots,d-1,
\end{align}

where $\kappa\neq 0$ can be viewed as a deformation parameter, since
formally, in the limit $\kappa\to\infty$ one obtains the commutative 
coordinate algebra of Minkowski space. Of course, to single out this 
limit as the Minkowski space requires some additional structure involving 
the action of the Poincar{\'e} group, or rather a deformed version 
thereof, for finite $\kappa$ \cite{kappafirst}. This was, originally, how the algebra was 
conceived and we shall return to this issue in Section~\ref{sec:4}. For the moment 
we concentrate on \eqref{kappa_Mink}. 

The first object of this paper is to discuss a class of star-products
on $\mathbb R^d$ and associated quantization maps based on the harmonic
analysis on the Lie group associated with \eqref{kappa_Mink}. The
motivation originates from a similar approach to the standard Weyl
quantization map based on its relation to the Heisenberg algebra of quantum
mechanics. For the purpose of later reference let us briefly recall
the main steps in this construction. The Heisenberg algebra associated
to a particle moving on the real line is the three-dimensional Lie
algebra defined by the relation
$$
[P,Q] = iC\,,
$$
where $C$ is a central element. The real form of this algebra with
basis $iP, iQ, iC$ has a faithful
representation $\sigma$ in terms of strictly upper triangular matrices:
\be\label{sigmarep}
\sigma(i(aP + bQ + cC)) = 
\left(\begin{array}{rrr} 0 & a & c\\ 0 & 0 & b\\ 0 & 0 & 0\end{array}\right)\,.
\ee
The connected and simply connected Lie group of the algebra is by
definition the \emph{Heisenberg group}, which we denote by $\CH
eis$. It is the group of upper triangular matrices of the form
\be\label{T}
T(a,b,c) = \left(\begin{array}{ccc} 1 & a & c +\frac 12 ab\\ 0 & 1 & b~~~\\ 0 & 0 & 1~~~\end{array}\right)\,,
\ee
which is obtained by exponentiation of \eqref{sigmarep}. The group
operations, expressed in this parametrization, are seen to be
\bq
T(a,b,c) T(a',b',c') &=& T(a+a',b+b',c+c'+\frac 12(ab'-a'b)), \\
T(a,b,c)^{-1} &=& T(-a,-b,-c) \,.
\eq
 It follows that $\CH eis$ is a unimodular group with Haar
measure equal to $dadbdc$. Thus the group algebra of $\CH eis$ can be
identified with ${\rm L}^1(\mathbb R^3)$ via the parametrization \eqref{T}. 
Let $\circ$ denote the convolution product on the group algebra.

According to the Stone-von Neumann theorem \cite{StoNeu} the non-trivial irreducible
unitary representations of $\CH eis$ are labelled by the value
$\hbar\neq 0$ of the central element $C$. Fixing $\hbar$, the representation $\pi$ can
be expressed in the form
\bq\label{heisrep}
\pi(T(a,b,c)) &=& e^{\frac i2 \hbar c} U(a,b)\,, \\
(U(a,b)\psi)(x) &=& e^{\frac i2 \hbar ab}e^{ibx}\psi(x-\hbar a)\,,
\eq
for $\psi \in{\rm L}^2(\mathbb R)$.
The corresponding representation of the group algebra, also denoted by
$\pi$, is then given by 
\be\label{red}
\pi(F) = \int_{\mathbb R^3} dadbdc F(a,b,c) \pi(T(a,b,c)) =
\int_{\mathbb R^2} F^\sharp (a,b)\,U(a,b)\,,
\ee 
where 
$$
F^\sharp(a,b) = \int dc F(a,b,c)e^{-i\hbar c}\,.
$$
Clearly, $F\to F^\sharp$ maps ${\rm L}^1(\mathbb R^3)$ onto  ${\rm L}^1(\mathbb R^2)$ 
and a simple calculation yields
\be\label{twist}
(F\circ G)^\sharp(a,b) = \int_{\mathbb R^2} da' db'
F^\sharp(a,b)G^\sharp(a-a',b-b')e^{\frac i2(ab'-a'b)}\,,
\ee
where the last expression is a ``twisted'' convolution
product on $\mathbb R^2$ that we shall denote by $F^\sharp\hat\circ
G^\sharp$, and where we have set $\hbar = 1$ for the sake of simplicity. 

According to \eqref{red} we may write $\pi(f)$ instead of $\pi(F)$
when $f=F^\sharp$. With this notation the Weyl quantization map $W$ is defined as
\be\label{weyl}
W(f) = \pi(\F f)\,,
\ee
for $f\in {\rm L}^1(\mathbb R^2)\cap \F^{-1}({\rm L}^1(\mathbb R^2))$,
where  $\F$ denotes the Fourier transform on $\mathbb R^2$,
\be\label{fourier}
 (\F f)(a,b) = \frac{1}{2\pi} \int d\alpha d\beta \, f(\alpha,\beta)
 e^{-i(a\alpha+b\beta)}. 
\ee
Using  
$$
\pi(F\circ G) = \pi(F)\pi(G)\qquad\mbox{for}\; F,G\in {\rm L}^1(\mathbb R^3)\,,
$$
we obtain from \eqref{red} and \eqref{twist} that 
$$
W(f {*_0} g) = W(f) W(g)\,,
$$
where $f$ and $g$ are functions on $\mathbb R^2$ and their \emph{Weyl-product} is given by
\be\label{weylprod}
f{*_0}g(\alpha,\beta) = \F^{-1}(\F f)\hat\circ (\F g))\,,
\ee
which clearly is well defined when $f$ and $g$ are Schwartz functions.

From this definition the familiar expressions (see e.g. \cite{victor}) for
the Weyl product can easily be derived. Likewise, the Weyl operators
$W(f)$ can be seen to be integral operators for appropriate functions
$f$. In particular, it can be shown that  $W(f)$ is of Hilbert-Schmidt
type if and only if $f$ is square integrable, and in this case 
\be\label{HSweyl}
\Vert W(f)\Vert^2_2 = 2\pi \int d\alpha d\beta |f(\alpha,\beta)|^2\,,
\ee
where $\Vert\cdot\Vert$ denotes the Hilbert-Schmidt norm.
It follows that the Weyl product can be extended to square integrable
functions and $W$ can be extended to an isomorphism between the
resulting algebra and the Hilbert-Schmidt operators on 
${\rm L}^2(\mathbb R)$. 

It is worth emphasizing that the construction outlined here depends
on the chosen para\-metrization \eqref{T}. An alternative
parametrization preserving the invariant measure is, e.g., 
$$(a,b,c)\to T(a,b,c+\xi ab),$$ 
where $\xi$ is a real constant.  
In this case,  one obtains a quantization map $W_\xi$
and a star product $*_\xi$ that are related to $*_0$ by
$$
W_\xi(f) = W(\Psi_\xi f)\,,\qquad \Psi_\xi (f*_{\xi}g)) = (\Psi_\xi f)*_0(\Psi_\xi g)\,,
$$
where $\Psi_\xi$ is defined by
$$
(\Psi_\xi f)(\a,\b) = e^{i\xi \a\b}f(\a,\b)\,.
$$
It follows that $\Psi_\xi$ is an isomorphism of star-algebras of
Schwartz functions and, moreover,
since both $\F$ and multiplication by a phase factor preserve the norm
in ${\rm L}^2(\mathbb R^2)$ we have that \eqref{HSweyl} is also fulfilled with $W$
replaced by $W_\xi$. In particular, one can verify that $W_{-\frac{1}{2}}$ 
is the so-called Kohn-Nirenberg quantization map, in which case
$W_{-\frac 12}(f)$ is the pseudo-differential operator with symbol
$f$. The Weyl map is singled our among the maps $W_\xi$ by the property
$$
W(f)^* = W(\bar f)\,,
$$
where $\bar f$ is the complex conjugate of $f$.

The purpose of this paper to is to develop an approach similar to the
preceding to quantization maps associated with $\kappa$-Minkowski space
for $d=2$, which we denote by $M_\kappa$. In Section~\ref{sec:2} we introduce the $\kappa$-Minkowski
group $G$, analogous to ${\CH}eis$, and via harmonic analysis on
$G$ we define a family of products, called \emph{star products}, and involutions for a
class $\CB$ of Schwartz functions on $\mathbb R^2$. Explicit
expressions for the star products and operator kernels
are obtained which are used to show that those involutive algebras have
natural isomorphic Banach algebra completions. In Section~\ref{sec:3}
two particular star products associated to the left and right
invariant Haar measures on $G$ are discussed. It is shown that they have 
natural extensions to a certain subalgebra  $\CC$ of the multiplier 
algebra of $\CB$ consisting of smooth functions of
polynomial growth. Moreover, it is shown that the resulting algebra
has a Hopf star algebra structure 
furnishing a star product representation of $\kappa$-Minkowski space. 
In Section~\ref{sec:4} we show how to represent
the  action of the  $\kappa$-Poincar{\'e} algebra ${\CP}_\kappa$ on
$\kappa$-Minkowski space in this particular realization as well as on
the subalgebra $\CB$. On the latter we show that the Lebesgue integral
is a twisted trace, invariant under the action of ${\CP}_\kappa$.
Finally, Section~\ref{sec:5} contains some concluding remarks and
a few technical details are collected in an appendix.

\section{Quantizations and star products}\label{sec:2}
\subsection{The right-invariant case}

In the following we restrict attention to $d=2$ in which case
the Lie algebra defined by \eqref{kappa_Mink} is the
unique noncommutative Lie algebra of dimension $2$ and 
$\kappa$-Minkowski space $M_\kappa$ is its universal enveloping algebra.  
We set $x=x_1$ and $t=\kappa x_0$ and consider the real form of the Lie algebra with
generators $it,ix$  fulfilling 

\be [t,x] = ix \label{kappa_1}. \ee
It  has a faithful $2$-dimensional representation $\rho$ given by

\be
\rho(it) = \twobytwo{-1}{0}{0}{0}, \;\;\; \rho(ix) = \twobytwo{0}{1}{0}{0},
\ee

and the corresponding connected and simply connected Lie group is the
group $G$ of $2\times 2$-matrices of the form 

\be\label{ab-param}
 S(a,b) = \twobytwo{e^{-a}}{b}{0}{1},\quad a,b\in \mathbb R, 
\ee
 obtained by exponentiating $\rho$: 
\be 
e^{i\rho(a t +  b' x)} = 
\twobytwo{e^{-a}}{\frac{1-e^{-a}}{a} b'}{0}{1}. 
\ee

The group operations written in the $(a,b)$ coordinates become

\be\label{group1}
S(a_1,b_1) S(a_2,b_2) = S(a_1+a_2, b_1 + e^{-a_1} b_2)\,, 
\qquad
S(a,b)^{-1} = S(-a, -e^a b)\,. 
\ee 

An immediate consequence is  

\begin{lemma}\label{haar}
The Lebesgue measure $da\, db$ is right invariant whereas 
the measure $ e^a da \, db$ is left-invariant on $G$. In particular,
$G$ is not unimodular.
\end{lemma} 

Let $\CA$ denote the convolution algebra of $G$ with respect to
the right invariant measure. Identifying functions on $G$ with
functions on $\mathbb R^2$ by the parametrization \eqref{ab-param} then
 $\CA$ is the involutive Banach algebra consisting of integrable functions on $\mathbb R^2$ with
product $\hat{*}$ and involution $^\dagger$ given by
\bq\label{starhat}
(f\,\hat{*}\, g) (a,b) &=& \int da'db'f(a-a',b-e^{a'-a}b')g(a',b')\,,\label{conv}\\
f^\dagger(a,b)&=& e^a\bar f(-a,-e^ab)\,,\label{dagger}
\eq
where $f,g\in\CA$ and $\bar f$ is the complex conjugate of $f$. 
If $\pi$ is a unitary representation  of $G$ (always assumed to be strongly
continuous in the following) it is well known (see e.g. \cite{GKP}) that $\pi$ gives rise to a
representation, also denoted by $\pi$, of $\CA$ by setting
\be
\pi(f) = \int da db \, f(a,b) \pi(S(a,b)). 
\ee
Thus, we have 
\be\label{starrep} 
\pi(f\,\hat{*}\, g) = \pi(f) \pi(g)\quad\mbox{and}\quad
\pi(f^\dagger) = \pi(f)^*\,. 
\ee

Following the same procedure as described for the Weyl quantization above we 
define the Weyl map $W_\pi$ associated with the representation $\pi$ by 
$$
W_\pi(f) = \pi(\F f)\;\;\mbox{for}\; f\in {\rm L}^1(\mathbb R^2)\cap \F^{-1}({\rm L}^1(\mathbb R^2))\,,
$$
where $\F$ denotes the Fourier transform \eqref{fourier} on
$\mathbb R^2$.
It then follows from \eqref{starrep} that

$$
W_\pi(f\,*\,g) = W_\pi(f)W_\pi(g)\quad\mbox{and}\quad W_\pi(f^*) =
W_\pi(f)^*$$

where the $*$-product and the $^*$-involution are
defined by

\be\label{star1}
 f\,*\,g =  \F^{-1} \left((\F f) \hat{*} (\F g) \right)\,.
\ee 
and

\be\label{inv1}
 f^* = \F^{-1}(\F(f)^\dagger)\,,
\ee
respectively. As in the case of the standard Moyal product, one needs to exercise care about the domain of
definition for the right-hand sides of \eqref{star1} and
\eqref{inv1}. I this section we restrict our attention to the subset $\CB$
of Schwartz functions introduced in the following definition, while an extension 
to a class of polynomially bounded functions will be discussed in subsequent sections.

\begin{definition} Let $\CS_c$ denote the space of Schwartz functions on $\mathbb
R^2$ with compact support in the first variable, i.e., 
$\mbox{supp}(f) \subseteq K\times\mathbb R$, where $K\subseteq \mathbb R$ is
compact. Then we define $\CB= \CF(\CS_c) = \CF^{-1}(\CS_c)$.
\end{definition}
\begin{prop}\label{alg1}
 If $f,g\in \CB$ then $f^*$ and $f*g$ also belong to $\CB$ and are given by
\be\label{star2}
f*g(\alpha,\beta) =  \frac{1}{2\pi} \int  dv \int d\alpha' \, f(\alpha+\alpha',\beta) g(\alpha, e^{-v} \beta) e^{-i \alpha' v},
\ee
and
\be\label{inv2}
f^*(\alpha,\beta) = \frac{1}{2\pi} \int dv\int d\alpha'\, 
\bar{f}(\alpha+\alpha', e^{-v} \beta) e^{-i\alpha'v}. 
\ee
respectively.
\end{prop}
\begin{proof}
Invariance of $\F^{-1}(\CS_c)$ under the $*$-product and the
$^*$-involution follows from the fact that $\CS_c$ is an
involutive subalgebra of the convolution algebra $\CA$ as is easily
seen from \eqref{starhat} and \eqref{dagger}. 

In order to establish \eqref{star2}, note that its right-hand side equals
\be\label{star3}
 \frac{1}{\sqrt{2\pi}} \int dv \tilde f(v,\beta)g(\alpha,e^{-v}\beta)e^{i\alpha v}\,,
\ee
where $\tilde f$ denotes the Fourier transform of $f$ w.r.t. the first
variable
$$
\tilde f(a,\beta) = \frac{1}{\sqrt{2\pi}} \int d\alpha \, f(\alpha,\beta)
e^{-ia\alpha}. 
$$
  Note that the
integrand in \eqref{star3} is a Schwartz function of $v,\alpha, \beta$
with compact support in $v$. Thus it suffices to show that the Fourier transform of
\eqref{star3} w.\,r.\,t. $\alpha,\beta$ equals $\F f\hat *\F g$. This follows 
from a straightforward calculation using the Plancherel theorem on 
the $\beta$-integral.

Concerning \eqref{inv2} we note similarly that the right-hand side equals 
\be\label{inv3}
 \frac{1}{\sqrt{2\pi}}\int dv \bar{\tilde f}(-v,e^{-v}\beta)e^{i\alpha v}\,.
\ee
Here, we note that the integrand is a Schwartz function of $v,\beta,$
such that Fourier transforming \eqref{inv3} w\,.r\,.t\,. $\beta$ gives
\be
 \frac{1}{\sqrt{2\pi}}\int dv \overline{\F f}(-v,-e^{v}b) e^v
 e^{i\alpha v} =  \frac{1}{\sqrt{2\pi}}\int dv ({\F f})^\dagger(v,b) e^{i\alpha v}\,.
\ee 
Hence, by Fourier inversion and \eqref{inv1} we conclude that the Fourier transform of the
right-hand side of \eqref{inv2} equals $\CF(f^*)$.  This proves \eqref{inv2}.
\end{proof}

Note that associativity of the above defined star product on
$\CB$  is an immediate consequence of associativity
of the convolution product on $\CA$. Likewise, $f\to f^*$  is an
involution on $\CB$, since $f\to f^\dagger$ is an
involution on $\CA$. Thus we have

\begin{corollary}\label{staralg}
$\CB$ equipped with the
$*$-product and $^*$-involution defined by \eqref{star2} and \eqref{inv2} is an involutive algebra.
\end{corollary}
   
It should be noted that the star product and the involution as defined 
by \eqref{star1} and \eqref{inv1} are independent of the choice of representation $\pi$
of $G$, while the quantization map $W_\pi$, that we proceed to discuss
next, is indeed representation dependent. $G$ being isomorphic to the
identity component of the group of affine transformations on $\mathbb R$, its representation
theory is well known \cite{GelNaj}. In particular, there is a close relationship
to the representation theory of the Heisenberg group \cite{Agostini}. 
The basic result we shall use is the following, the proof of which is included for 
the sake of completeness (see also \cite{DabPia1}).  

\begin{prop}\label{rep1}
$G$ has exactly two non-trivial unitary
representations $\pi_\pm$. Their action on the generators $t,x$ is
given by
\bq\label{Pi+-}
\pi_+(t) &=& -i\frac{d}{ds}\,,\qquad \pi_+(x) = e^{-s}\label{rep+}\,,\\
\pi_-(t) &=& -i\frac{d}{ds}\,,\qquad \pi_-(x) = -e^{-s}\label{rep-}\,,
\eq
as self-adjoint operators on ${\rm L}^2(\mathbb R)$.

All other irreducible unitary representations are one-dimensional of
the form $\pi_c(x)=0$ and $\pi_c(t)=c$ for some $c\in\mathbb R$.
\end{prop}

\begin{proof} Let $\pi$ be a unitary representation of $G$ on a
  Hilbert space $H$ and let $v\in Dom(\pi(x))$. Differentiating the
  relation
\be\label{expid}
e^{i\a\pi(t)}e^{ib\pi(x)}e^{-ia\pi(t)}v = e^{ib e^{-a}\pi(x)}v\,,
\ee
which follows from \eqref{group1}, w.r.t. $b$ we get that $e^{-ia\pi(t)}v\in Dom(\pi(x))$ and 
$$
e^{ia\pi(t)}\pi(x)e^{-ia\pi(t)}v = e^{-a}\pi(x)v\,,
$$
so the two self-adjoint operators
$e^{ia\pi(t)}\pi(x)e^{-ia\pi(t)}$ and  $e^{-a}\pi(x)$ 
coincide. But since $e^{-a}>0$ the spectral subspaces $H_+, H_-$ and
$H_0$ corresponding to the positive and negative real line and $\{0\}$,
respectively, are identical for $\pi(x)$ and
$e^{ia\pi(t)}\pi(x)e^{-ia\pi(t)}$. It follows that those spaces are invariant
  under $e^{ia\pi(t)}$ and  $e^{ib\pi(x)}$. By irreducibility one
  of them equals $H$ and the other two vanish. 

Assume $H=H_+$ and define the self-adjoint operator $Q$ by
$$
Q = -\ln(\pi(x))\,.
$$
Then $x=e^{-Q}$ and by \eqref{expid} we have
$$
\exp\left(i b e^{ia\pi(t)}e^{-Q}e^{-ia\pi(t)}\right) = \exp\left(i b e^{-a}e^{-Q}\right)\,,
$$
and hence 
$$
 e^{ia\pi(t)}e^{-Q}e^{-ia\pi(t)} = \exp\left(e^{ia\pi(t)}Qe^{-ia\pi(t)}\right) =  e^{-Q-a}\,.
$$
Taking logarithms gives
$$
e^{ia\pi(t)}Qe^{-ia\pi(t)} =  Q+a
$$
and consequently
$$
e^{ia\pi(t)}e^{i b Q}e^{-ia\pi(t)} = e^{iab}e^{i b Q}\,.
$$
which is recognized as the Weyl form of the canonical commutation
relations. Applying the Stone-von Neumann theorem \cite{StoNeu} we conclude that
$\pi=\pi_+$. 
Similarly one shows that $\pi=\pi_-$ if $H=H_-$, and the case $H=H_0$
yields the one-dimensional representations as asserted.
\end{proof}

We will use the notation $W_\pm$ for $W_{\pi_\pm}$.
 From the explicit form \eqref{heisrep} of the action of the Heisenberg group in an
irreducible representation one obtains the action of $G$ in the representations
$\pi_\pm$. Using $S(a,b)=S(0,b)S(a,0)$ the result is
$$
\pi_\pm(S(a,b))\psi(s) = e^{\pm ibe^{-s}}\psi(s+a)\,,\quad \psi\in{\rm L}^2(\mathbb R)\,.
$$
It is now straightforward to determine the action of
$W_\pm(f)$ for arbitrary $f\!\in\!{\rm L}^1(\mathbb R^2)\cap \F^{-1}(
{\rm L}^1(\mathbb R^2))$. If $\langle\varphi ,\psi\rangle$ denotes the
inner product of $\varphi,\psi\in {\rm L}^2(\mathbb R^2)$ we get
\bq
\langle\varphi,W_\pm(f)\psi\rangle &=& \int dadbds\; \F f(a,b)\, \overline\varphi(s)\, e^{\pm
  ibe^{-s}}\psi(s+a)\nonumber\\
&=&  \int dsdudb\; \overline\varphi(s)\F f(u-s,b)\, e^{\pm i b e^{-s}}\psi(u)\nonumber\\
&=& \sqrt{2\pi} \int dsdu\; \overline\varphi(s)\tilde f(u-s,\pm
e^{-s})\psi(u)\,.
\enq
Hence we have shown 
\begin{prop}\label{intop}
 For  $f\in {\rm L}^1(\mathbb R^2)\cap \F^{-1}(
{\rm L}^1(\mathbb R^2))$ the operators $W_\pm(f)$ are integral
operators on ${\rm L}^2(\mathbb R)$  with kernels given by
$$
K^\pm_f(s,u)=  \sqrt{2\pi} \tilde f(u-s,\pm
e^{-s}) = \int dv f(v,\pm e^{-s})\, e^{-iv(u-s)}\,.
$$ 
\end{prop}

As a consequence we can establish the following basic identities.

\begin{prop}\label{trace1}
\begin{itemize}
\item[a)] $W_\pm(f)$ is of Hilbert-Schmidt type if and only if the
  restriction of  $f$ to $\mathbb R\times\mathbb R_\pm$ is square
  integrable w.r.t. the measure 
$$
d\mu = |\beta|^{-1}d\alpha d\beta \,,
$$
and we have
\be\label{HS}
\Vert W_\pm(f)\Vert_2^2 = 2\pi\int_\mathbb R da \int_{\mathbb R_\pm}db
|f(\alpha,\beta)|^2\frac{d\alpha d\beta}{|\beta|} = 2\pi\int_{\mathbb R^2}dsdv |f(v,\pm e^{-s})|^2\,,
\ee
where $\Vert\cdot\Vert_2$ denotes the Hilbert-Schmidt norm.
\item[b)] If $W_\pm(f)$ is trace class then
\be\label{trace}
tr W_\pm(f) = \int_{\mathbb R^2}dsdv f(v,\pm e^{-s})\,.
\ee
\end{itemize}
\end{prop}

\begin{proof} 
\noindent a)~ The operator  $W_\pm(f)$ is Hilbert-Schmidt if and only
if its kernel is square integrable. From Proposition~\ref{intop} we get
$$
\int dsdu |K^\pm_f(s,u)|^2 = 2\pi \int dsdu |\tilde f(u-s,\pm
e^{-s})|^2 =  2\pi \int dsdu |\tilde f(u,\pm e^{-s})|^2\,.
$$
Applying the Plancherel theorem on the $u$-integral then proves the
first assertion as well as \eqref{HS}.

\noindent b)~ If $W_\pm(f)$ is trace class, then
$$
tr  W_\pm(f) = \int_{\mathbb R^2}ds K^\pm_f(s,s)\,,
$$
and \eqref{trace} follows from Proposition~\ref{intop}.
\end{proof}

We note that although $\CB$ is not contained in ${\rm L}^2(\mathbb
R^2,d\mu)$ we have that $\CB\cap{\rm L}^2(\mathbb R^2,d\mu)$ is
dense in ${\rm L}^2(\mathbb R^2,d\mu)$. Indeed, let $\CB'$ denote
the subspace of $\CB$ consisting of Fourier transforms of
derivatives w.\,r.\,t. the second variable of functions in $\CS_c$.
A function $f(\alpha,\beta)$ in $\CB'$ is then of the form $\beta
g(\alpha,\beta)$ where $g$ is a Schwartz function, hence 
$f\in{\rm L}^2(\mathbb R^2,d\mu)$. Moreover, if $f$ is orthogonal to $\CB'$ in 
${\rm L}^2(\mathbb R^2,d\mu)$ then its Fourier transform,
considered as a tempered distribution, vanishes as a
distribution, hence also as a tempered distribution. Thus $f=0$ and we
conclude that $\CB'$ is dense in ${\rm L}^2(\mathbb R^2,d\mu)$.
 
It follows from this remark and \eqref{HS} that the mappings
$W_\pm$ have unique extensions from $\CB'$ to 
${\rm L}^2(\mathbb R^2,d\mu)$ such that \eqref{HS} still holds. In particular, the map 
$$W: f\to W_+(f)\oplus W_-(f)$$
is injective from ${\rm L}^2(\mathbb R^2,d\mu)$ into $\CH\oplus\CH$, where $\CH$ 
denotes the space of Hilbert-Schmidt operators on ${\rm L}^2(\mathbb R)$.  

On the other hand, it is clear from the proof of Proposition~\ref{trace1} 
that any pair of kernels $K^\pm$ in ${\rm L}^2(\mathbb R^2)$ originate from 
an $f\in {\rm L}^2(\mathbb R^2,d\mu)$, i.e. $W$ is unitary up to a factor $\sqrt{2\pi}$. 
This proves the following extension result.
\begin{theorem}\label{Bbar}
Let $\CB'$ and $W$ be as defined above and set $\bar\CB= {\rm L}^2(\mathbb R^2,d\mu)$. Then the
$*$-product \eqref{star2} and involution \eqref{inv2} have unique extensions
from $\CB'$ to $\bar\CB$, such that $\bar\CB$ becomes a Banach algebra and $W$ 
an isomorphism,
$$
W(f*g) = W(f)W(g)\qquad  W(f^*) = W(f)^*\,.
$$
If we complete the algebra $\bar\CB$ in the operator norm, the 
resulting $C^*$ algebra will be that of compact operators.  
\end{theorem}
\begin{corollary}
The integral w.\,r.\,t. $d\mu$ over $\mathbb R\times \mathbb R_\pm$ is a
positive trace on $\bar\CB$ in the following sense: for any $f,g\in \bar\CB$,
$$
\int duds\, (f*f^*)(u,\pm e^{-s}) \geq 0\quad\mbox{and}\quad \int duds (f*g)
(u,\pm e^{-s})  = \int duds (g*f) (u,\pm e^{-s})\,.
$$
\end{corollary}
\begin{proof}
If $f\in\bar\CB$ then $W(f)$ is Hilbert-Schmidt and the first
inequality follows from \eqref{HS}. If $f,g\in\bar\CB$ then $W(f)W(g)$
is trace class and the second identity follows from Theorem~\ref{Bbar} and \eqref{trace}.
\end{proof}

For later use we note the following identities.

\begin{prop}
\begin{itemize}
\item[a)]~ If  $f,g\in\bar\CB$ then
$$
\int d\alpha d\beta |\beta|^{-1} (f*g^*)(\alpha,\beta) = \int d\alpha
d\beta |\beta|^{-1}f(\alpha,\beta)\,\bar g(\alpha,\beta)\,.
$$
\item[b)]~ If $f,g\in \CB$ then
\bq
 \int d\alpha d\beta \, (f*g^*)(\alpha,\beta) &=& 
\int d\alpha d\beta \, f(\alpha,\beta)\,\bar
g(\alpha,\beta)\,,\label{isom}\\
 \int d\alpha d\beta
f^*(\alpha,\beta) &=& \int d\alpha d\beta \bar f(\alpha,\beta)\,.\label{invint}
\eq
\end{itemize}
\end{prop}

\begin{proof}
a)~ Follows immediately from Proposition~\ref{trace1} and
Theorem~\ref{Bbar}.

b)~ Using \eqref{star1} and \eqref{inv1} as well as \eqref{conv} and
\eqref{dagger} we have 
\bq\nonumber
 \int d\alpha d\beta \, (f*g^*)(\alpha,\beta) &=& \F(f*g^*)(0) =
 (\F(f)\hat *\F(g^*))(0)\\ &=&  (\F(f)\hat
 *\F(g)^\dagger)(0) = \int da db \, \F f(-a,-e^ab)\overline{\F g}(-a,-e^ab)e^a\nonumber\\ &=&
 \int da db \, \F f(a,b)\overline{\F g}(a,b) =  \int d\alpha d\beta \,
 f(\alpha,\beta)\,\bar g(\alpha,\beta)\,.\nonumber
\eq
Similarly, we have
$$
 \int d\alpha d\beta \, f^*(\alpha,\beta) = \F(f^*)(0) =
 (\F(f)^\dagger)(0) =  \overline{\F(f)}(0) =
 \int d\alpha d\beta \, \bar f(\alpha,\beta)\,.
$$
\end{proof}

In particular, it follows that  
$$
\int d\alpha d\beta \, (f*f^*)(\alpha,\beta) \geq 0\,,\quad f\in\CB\,,
$$
 but in general $\int d\alpha d\beta\, f*g(\alpha,\beta)
\neq \int d\alpha d\beta\, g*f(\alpha,\beta)$, i.e.  $\int d\alpha
d\beta $ is not a trace on $\CB$. However, we shall see in 
Proposition~\ref{invintgen} that $\int d\alpha d\beta$ satisfies a
twisted trace property.

\subsection{The left-invariant case and other star-products}

The above procedure can be also applied to the convolution algebra 
of the left invariant measure on $G$ instead of the right invariant 
one. It is then convenient to use the parametrization
\be\label{R}
R(a,c) = S(a,e^{-a}c)\,,\quad a,c\in\mathbb R\,,
\ee
in which the left invariant measure is $dadc$ by
Lemma~\ref{haar}. Given a unitary representation $\pi$ of $G$, the
corresponding quantization map $\tilde W_\pi$ is defined by
$$
\tilde W_\pi(f) = \int da dc\, \F f(a,c)\pi(R(a,c)) = \int dadc\, \F f(a,c)\pi(S(a,e^{-a}c))\,,
$$ 
for $f\in {\rm L}^1(\mathbb R^2)\cap \F^{-1}({\rm L}^1(\mathbb R^2))$.   
More generally, let us consider the map $W^\varphi_\pi$ given by
\be\label{Wphi1}
W^\varphi_\pi(f) = \int da dc\, \F f(a,c) \pi(S(a,\varphi(a)c))\,,
\ee
where $\varphi$ is a smooth, positive function on $\mathbb R$. Defining  
$$
(Uf)(a,b) = f\left(a,\eta(a)b\right)\eta(a)\,,
$$
for any function $f$ of two variables, where
$$
\eta (a) = \varphi(a)^{-1}\,,
$$
a change of variables in \eqref{Wphi1} gives 
\be\label{Wphi2}
W^\varphi_\pi(f) = \pi(U\F f)\,.
\ee
The corresponding star-product $*_\phi$ and involution
$^{*_\varphi}$ are given by 
\be\label{starphi1}
f*_\varphi g (\alpha,\beta) = 
\frac{1}{2\pi}\F^{-1}U^{-1}((U\F f)\,\hat{*}\,(U\F g))
\ee
and 
\be\label{invphi1}
f^{*_\varphi}(\alpha,\beta) =  \F^{-1} U^{-1}((U\F f)^\dagger)\,,
\ee  
which are easily seen to be well defined for $f,g\in \CB$. 
More explicitly, the following result holds.
\begin{prop}\label{starphi}
If $f,g\in\CB$ and $\varphi$ is positive and smooth then
\be\label{starphi2}
f*_\varphi g (\alpha,\beta) = \frac{1}{2\pi}\int da db \,
\tilde{f}(b,\omega(a,b)e^{a-b}\beta) \tilde{g}(a-b,\omega(a,a-b)\beta) e^{i\alpha a}\,,
\ee
and 
\be\label{invphi2}
f^{*_\varphi}(\alpha,\beta) = \frac{1}{2\pi}\int dv\int d\alpha' \,
\bar{f}(\alpha +\alpha',\omega(a,-a)e^a\beta)e^{-i\alpha'v}\,,
\ee
where
$$
\omega(a,b) = \eta(a)\varphi(b)e^{b-a}\,.
$$
In particular, the star product $\star$ for the left-invariant measure, 
obtained for $\varphi(a)=e^{-a}$, becomes
\be\label{leftstar}
f\star g (\alpha,\beta) = \frac{1}{2\pi}\int dv\int d\alpha' \, 
f(\alpha,e^v\beta) g(\alpha +\alpha',\beta) e^{-i\alpha'v}\,,
\ee
and the involution $^\star$ for the left-invariant product is
\be\label{leftinv}
f^{\star}(\alpha,\beta) = \frac{1}{2\pi}\int dv\int d\alpha' \,
\bar{f}(\alpha+\alpha',e^v\beta)e^{-i\alpha'v}\,.
\ee
\end{prop}

\begin{proof}
The first two identities follow by straightforward computation using
\eqref{starphi1} and \eqref{invphi1} and Fourier inversion. The last
two identities follow from the first two after a change of variables
combined with Fourier inversion. Details are left to the reader.
\end{proof}

\begin{definition} By $\CB_\varphi$ we shall denote the involutive
  algebra obtained by equipping $\CB$ with the
  product $*_\varphi$ and involution $^{*_\varphi}$. 
\end{definition}

\begin{remark}~ In \cite{Andrea} a star product is obtained by a
  somewhat different approach involving a reducible representation of
  $G$ acting on functions of two variables. Although the explicit
  form of that star product is not given in \cite{Andrea}, it can be 
  verified that it indeed coincides with \eqref{star2}. 

The star product considered in \cite{DabPia1} (and in \cite{DabPia2, DabPia3, DabPiaGod}) corresponds to the 
case $\varphi(a)= \frac{1-e^{-a}}{a}$ above and has the 
property that the involution equals complex conjugation. However, this property
does not determine the star product uniquely among the products
$*_\varphi$, as it holds more generally if $\varphi$ satisfies the relation
$$ \varphi(-a) = e^a\varphi(a)\,,\quad a\in\mathbb R\,. $$
\end{remark}

\bigskip

The form of the Weyl operators $W^\varphi_\pm(f)$ for $\pi=\pi_\pm$ is obtained
from \eqref{Wphi2} and Proposition~\ref{intop} by an easy computation
that we omit. The result is the following.

\begin{prop}\label{intop2} Assume $\varphi$ is positive and smooth.
 For  $f\in {\rm L}^1(\mathbb R^2)\cap \F^{-1}(
{\rm L}^1(\mathbb R^2))$ the operators $W^\varphi_\pm(f)$ are integral
operators on ${\rm L}^2(\mathbb R)$  with kernels given by
$$
K^\pm_f(s,u)=  \sqrt{2\pi} \tilde f(u-s,\pm\varphi(u-s)
e^{-s}) = \int dv f(v,\pm \varphi(u-s)e^{-s}) e^{-iv(u-s)}\,.
$$ 
\end{prop}

It can now be seen that the norm and trace formulas \eqref{HS} and
\eqref{trace} hold independently of the choice of $\varphi$.

\begin{prop}\label{trace2}
\begin{itemize}
\item[a)] $W^\varphi_\pm(f)$ is Hilbert-Schmidt if and only if the
  restriction of  $f$ to $\mathbb R\times\mathbb R_\pm$ is square
  integrable w.r.t. the measure $d\mu$ and we have
\be\label{HSphi}
\Vert W^\varphi_\pm(f)\Vert_2^2 =2\pi \int_{\mathbb R^2}dsdv |f(v,\pm e^{-s})|^2\,.
\ee
\item[b)] If $W^\varphi_\pm(f)$ is trace class then
\be\label{tracephi}
tr W^\varphi_\pm(f) = \int_{\mathbb R^2}dsdv f(v,\pm e^{-s})\,.
\ee
\end{itemize}
 \end{prop}
\begin{proof}
a)~ Using Proposition~\ref{intop2} we get
\bq\nonumber
\Vert W^\varphi_\pm(f)\Vert_2^2 &=& 2\pi \int duds|\tilde
f(u-s,\pm\varphi(u-s)e^{-s})|^2\\ &=& 2\pi\int dvds  f(v,\pm\varphi(v)
e^{-s})|^2 = 2\pi\int dvdr| f(v,\pm e^{-r})|^2
\enq 
which coincides with \eqref{HS}.

b)~ Similarly, we have
$$
tr W^\varphi_\pm(f) = \sqrt{2\pi} \int ds  \tilde f(0,\pm\varphi(0) e^{-s}) = \int
dvds f(v,\pm e^{-s})
$$
as claimed.
\end{proof}

\begin{definition}
By $\bar\CB_\varphi$ we denote the Banach algebra obtained by
equipping ${\rm L}^2(\mathbb R^2, d\mu)$ with the product and involution
defined by \eqref{starphi2} and \eqref{invphi2} and extended from
$\CB'$ using \eqref{HSphi} in the same manner as for the case $\varphi=1$
treated previously.
\end{definition}

\begin{theorem}\label{iso}
The involutive algebras $\CB_\varphi$, resp. $\bar\CB_\varphi$, where $\varphi$
is positive and smooth, are isomorphic.
\end{theorem}

\begin{proof}
For the case of $\CB_\varphi$ we note that $\F^{-1}U\F$ maps
$\CB_\varphi$ onto $\CB=\CB_{\varphi=1}$ and is by construction a
homomorphism by
\eqref{star1},\eqref{inv1},\eqref{starphi1} and \eqref{invphi1}. The inverse map is 
obtained by replacing $\varphi$ by $\eta$.

The same argument applies to $\bar\CB_\varphi$ since Theorem~\ref{trace2}
shows that $\F^{-1} U\F$ is an isometry on $\CB'$ and
therefore its extension is an isometry from $\CB_\varphi$ onto $\CB$. 
\end{proof}

\begin{remark} The quantization maps $W_\pm$ were also considered in
  \cite{Agostini} and relations \eqref{HS} and \eqref{trace} were
  likewise derived.

For the particular case $\varphi(a) = \frac{e^a-1}{a}$, relations
\eqref{HSphi} and \eqref{tracephi} also appear in \cite{DabPia1}.
\end{remark}

\section{Hopf algebra structure}\label{sec:3}

The star algebras $\CB_\varphi$ or $\bar\CB_\varphi$ defined in the previous
section obviously do not contain the coordinate functions $\alpha$ and
$\beta$. Hence, to obtain a representation of $M_\kappa$ with $\alpha$ and 
$\beta$ representing the generators $t$ and $x$ in \eqref{kappa_1} we need an 
extension of the domain of definition for the star product and involution. 
It is the purpose of this section to exhibit such an extension. 

As originally mentioned in \cite{kappafirst} and developed in \cite{MajRue}, 
$M_\kappa$ has a natural Hopf algebra structure, which arises by dualization 
of the momentum subalgebra of the $\kappa$-Poincar\'e Hopf algebra. The coalgebra
structure $(\cop,\varepsilon)$ and antipode $S$ are defined by
$$
\begin{aligned} 
&\cop t = t \ts 1 + 1 \ts t\,, \;\;\; \cop x = x \ts 1 + 1 \ts x\,, \\
&\varepsilon(t)=\varepsilon(x)=0\,,\\
&S(t) = -t\,,\quad S(x) = -x\,.
\end{aligned}
$$
As will be seen in Theorem~\ref{iota} below the extension we present allows a realization of the
full Hopf algebra structure of $\kappa$-Minkowski space.
Unless stated explicitly otherwise we work with the star product
associated with the right invariant measure because of its simple 
form (\ref{star2}). Analogous results for the $\star$-product
\eqref{leftstar} are obtained similarly.

\subsection{The algebra $\CC$}

Using standard notation
$\partial_\alpha^n=\prod_{i=1}^k\frac{\partial}{\partial{\alpha_i}}$
for $\alpha=(\alpha_1,\dots,\alpha_k)\in\mathbb R^k\,,
n=(n_1,\dots,n_k)\in\mathbb N_0^k$, $\mathbb N_0= \{0, 1, 2,
3,\dots\}$, and with $|\cdot |$ denoting the Euclidean norm on
$\mathbb R^k$ we introduce the following function spaces.  

\begin{definition}\label{defC}
Let $\CC_k$ be the space of smooth functions $f(\alpha,\beta)$ on $\mathbb R^{2k}$ satisfying
polynomial bounds of the form
\be\label{pol}
|\partial_\alpha^n\partial_\beta^m f(\alpha,\beta)|\;\leq\;
c_{n,m}(1+|\alpha|)^{N_n}(1+|\beta|)^{M_{n,m}}\,,
\ee 
for all $\alpha,\beta\in\mathbb R^k$  and such that the Fourier
transform $\tilde f$ of $f$ (as a tempered distribution) w.r.t. 
$\alpha$ has compact support in $\alpha$. Here, $n,m \in {\mathbb N}^k$ 
are arbitrary and $N_n, M_{n,m}$ are constants of which the former 
is independent of $m$, and $c_{n,m}$ is a positive constant.

Given $f\in \CC_k$ we denote by $K_f$ the smallest compact subset of
$\mathbb R^k$ such that $\hbox{supp}(f)\subseteq K_f\times \mathbb R^k$ and we
call $K_f$ the $\alpha$-support of $f$.

For $k=1$ we set $\CC=\CC_1$ and we have the canonical 
inclusion $\CC \otimes \CC \hookrightarrow \CC_2$ given by:
$$
(f\ts g)(\alpha_1,\alpha_2,\beta_1,\beta_2) = f(\alpha_1,\beta_1)
g(\alpha_2,\beta_2),
$$
for $f,g\in\CC$.
\end{definition}

\begin{remark}
Note that $\CB \subseteq \CC$ and, additionally, if $f\in \CC$ and $p$ is 
a polynomial in $\alpha,\beta$, then $p$ and $pf$ are in $\CC$.
\end{remark}

In order to extend the $*$-product to $\CC$ let $f,g\in\CC$ and define
for fixed $\alpha, \beta\in\mathbb R$,
\be
g_{\alpha,\beta}(v) = g(\alpha, e^{-v}\beta)e^{i\alpha v}\,,\quad v\in
\mathbb R\,.
\ee
Motivated by \eqref{star3} we then set  
\be\label{star4}
(f\,*\,g)(\alpha,\beta) = \frac{1}{\sqrt{2\pi}}\int dv \tilde f(v,\beta) g_{\alpha,\beta}(v)
\ee
which is well-defined since $\tilde f(v,\beta)$ has compact support in
$v$ and $g_{\alpha,\beta}$ is a smooth function.
We show below that $f * g\in \CC$  and that
\be\label{addsupp}
K_{f*g} \subseteq K_f + K_g\,.
\ee
In fact, viewing the product \eqref{star3} as a linear map on $\CC\ts\CC$,
it extends to a linear map $m_*:\CC_2\to\CC$ by setting
\be\label{genprod}
\begin{aligned}
(m_*F)(\alpha,\beta) &= \frac{1}{2\pi}\int d\alpha'\int dv\, \chi^1_F(v)
F(\alpha',\alpha,\beta,e^{-v}\beta)e^{i(\alpha-\alpha')v}\\
& = \frac{1}{2\pi}\int d\alpha'\int dv\, \chi^1_F(v)
F(\alpha+\alpha',\alpha,\beta,e^{-v}\beta)e^{-i\alpha'v}\,,
\quad F \in \CC_2\,,
\end{aligned}
\ee
where $\chi^1_F$ denotes a smooth function on $\mathbb R$ of
compact support such that 
$$\chi_F^1(v_1)\tilde F(v_1,v_2,\beta_1,\beta_2) =
\tilde F(v_1,v_2,\beta_1,\beta_2)$$
 as distributions, that is $\chi_F^1$
equals $1$ on a neighborhood of the projection of $K_F$ on the 
first axis. Note that \eqref{genprod} coincides with \eqref{star4} 
if $F=f\ts g, f,g\in
\CC$. Convergence of the double integral in \eqref{genprod}
is a consequence of the polynomial bounds \eqref{pol} for $F$, as can
be seen as follows. Let $\zeta$ be a smooth function of compact
support on $\mathbb R$ that equals $1$ on a neighborhood of $0$ and
write the integral in \eqref{genprod} as a sum of two terms $F_1(\alpha,\beta)$ 
and $F_2(\alpha,\beta)$ where 
\be\label{F1}
F_1(\alpha,\beta) =  \frac{1}{2\pi}\int d\alpha'\int dv\, \zeta(\alpha-\alpha')\chi_F(v)
F(\alpha',\alpha,\beta,e^{-v}\beta)e^{i(\alpha-\alpha')v}\,.
\ee
Obviously, this latter integral is absolutely convergent and by
repeated differentiation w.\,r.\,t. $\alpha, \beta$ it is seen that $F_1$
is smooth and satisfies polynomial bounds of the form \eqref{pol}. 
For $F_2$, given by formula \eqref{F1} with $\zeta$ replaced by
$1-\zeta$, one obtains after $N$ partial integrations w.\,r.\,t. $v$
the expression
\be
\begin{aligned}
&F_2(\alpha,\beta) =  & \\
&= \frac{i^N}{2\pi}\int \! d\alpha' \! \int \!dv \,
(\alpha \!-\!\alpha')^{-N}(1\!-\!\zeta(\alpha\!-\!\alpha'))\frac{\partial^N}{\partial v^N}\left(\chi_F(v)
F(\alpha',\alpha,\beta,e^{-v}\beta)\right)e^{i(\alpha-\alpha')v}\,.&
\end{aligned}
\ee
Choosing $N$ large enough one obtains an absolutely convergent integral
as a consequence of \eqref{pol}, using that $N_n$ is independent of $m$. Applying 
the same argument to derivatives of the integrand w.\,r.\,t. $\alpha, \beta$ it follows 
easily that $F_2$ is smooth and satisfies the bounds \eqref{pol}. In the Appendix
it is proven that $m_*F$ is independent of the choice of $\chi^1_F$
with the mentioned property and that
\be\label{supprod}
\mbox{supp}({m_*F})\subseteq \{v_1+v_2\mid (v_1, v_2)\in K_F\}\times\mathbb R\,,
\ee
of which \eqref{addsupp} is a special case.
In particular, $m_*(F)$ has compact $\alpha$-support and hence we may conclude
that $m_*(F) \in \CC$ for all $F \in \CC_2$. 

More generally, we can define maps $\CC_{k+1}\to\CC_k$ by
letting $m_*$ act on any two pairs of variables
$(\alpha_i,\beta_i),(\alpha_j,\beta_j)$  while keeping the 
other variables fixed. We shall use the notation $m_*\ts 1$ and $1\ts m_*$ for the
maps $\CC_3\to\CC_2$ where $m_*$ acts on
$(\alpha_1,\beta_1),(\alpha_2,\beta_2)$ and
$(\alpha_2,\beta_2),(\alpha_3,\beta_3)$, respectively. Using arguments
similar to those above one proves associativity of $m_*$, that is
\be\label{genass}
m_*(m_*\ts 1)= m_*(1\ts m_*)\,.
\ee
Details are given in the appendix.

Next, we proceed to define the involution on $\CC$. A convenient form is 
obtained from \eqref{inv3} which, after a simple change of variables, yields 
\be\label{inv4}
\int d\alpha d\beta\, f^*(\alpha,\beta)\,\tilde\phi(\alpha,\beta) = \int dv
d\beta \,\bar{\tilde f}(v,\beta)\,
\chi_f(v)\,\phi(-v,e^{-v}\beta)\,e^{-v}\,,
\ee
for $\phi\in\CS(\mathbb R^2)$, where $\chi_f$ is an  arbitrary smooth function of compact
 support that equals $1$ on a neighborhood of $K_f$. Defining
\be\label{R-op}
(R_f\phi)(v,\beta)  = \chi_f(v)\,\phi(-v,e^{-v}\beta)\,e^{-v}\,,
\ee
for  $\phi\in\CS(\mathbb R^2)$, it is clear that
$R_f$ is a continuous mapping from $\CS(\mathbb R^2)$ into itself. Hence, it follows that
\eqref{inv4} defines $f^*$ as a tempered distribution for any
$f\in\CC$. We refer to the appendix for a proof that $f^*$ is
independent of the choice of function $\chi_f$ with the asserted property. The
Fourier transform of $f^*$ w.\,r.\,t. $\alpha$ is given by
\be\label{inv5}
\widetilde{f^*}(\phi) = \bar{\tilde f}(R_f\phi),
\ee 
from which it is clear that the $\alpha$-support of $f^*$ fulfils
\be\label{suppstar}
K_{f^*} \subseteq - K_f\,.
\ee
Hence we can choose $\chi_{f^*}(v)=\chi_f(-v)$. It is then easily
verified that
$$
R_f R_{f^*}\phi = \chi_f^2\, \phi
$$
which by use of  \eqref{inv5} gives 
$$
\widetilde{f^{**}}(\phi) = \bar{\widetilde{f^*}}(R_{f^*}\phi)
= \tilde f(R_fR_{f^*}\phi) = \tilde f(\phi)\,,
$$
since $\chi_f^2$ equals $1$ on a neighbourhood of $K_f$. This shows that
$$
f^{**} = f\,,\qquad f\in\CC\,.
$$
In order to show that $f^*\in\CC$ for any $f\in\CC$ we first note that
$f^*$ is, in fact, a function given by the following generalization of \eqref{inv2},
\be\label{geninv}
f^*(\alpha,\beta) = \frac{1}{2\pi}\int d\alpha'\int dv\, \chi_f(-v)\, \bar f(\alpha +
\alpha',e^{-v}\beta)\,e^{-i\alpha' v}\,.
\ee
Indeed, convergence of this double integral is a consequence of \eqref{pol},
which can be seen by arguments similar to those for $m_*$ as
follows. Let again $\zeta$ be a smooth function on
$\mathbb R$ of compact support that equals $1$ on a neighborhood of
$0$ and write the integral in \eqref{geninv} as a sum of two terms
$f^*_1(\alpha,\beta)$ and $f^*_2(\alpha,\beta)$, where
\be\label{f1star}
f^*_1(\alpha,\beta) = \frac{1}{2\pi}\int d\alpha'\int dv \, \zeta(\alpha')\chi_f(-v) \bar f(\alpha +
\alpha',e^{-v}\beta)e^{-i\alpha' v}\,.
\ee
Clearly, this latter integral is absolutely convergent and by repeated
differentiation w.\,r.\,t. $\alpha,\beta$ it follows easily that
$f^*_1$ is smooth and satisfies \eqref{pol}. For $f^*_2$ one obtains after 
$N$ partial integrations w.\,r.\,t.  $v$ the expression
\be\label{f2star}
f^*_2(\alpha,\beta) =   \frac{(-i)^N}{2\pi} \int d\alpha' \int dv \,
(1-\zeta(\alpha'))\frac{\partial^N}{\partial v^N}\left(\chi_f(-v) \bar f(\alpha +
\alpha',e^{-v}\beta)\right)e^{-i\alpha' v}\,.
\ee
By choosing $N$ large enough this integral is absolutely convergent by
\eqref{pol}. Applying the same argument to arbitrary derivatives of
the integrand w.\,r.\,t. $\alpha,\beta$ it follows easily that
$f^*_2$ is smooth and satisfies \eqref{pol}. Having proven convergence
of the integral \eqref{geninv}, its
coincidence with $f^*$ follows easily. Hence, we have shown that
$f^*$ belongs to $\CC$ and is given by \eqref{geninv}.

Evidently, the preceding arguments can be extended to show that, more generally, an
involution $^*$ on $\CC_k$ is obtained by setting
\be\label{geninv-k}
\begin{aligned}
F^*(\alpha,\beta) =& \\
&= (2\pi)^{-k}\int d^k\alpha'\int d^kv \chi_F(-v) \bar F(\alpha +
\alpha',e^{-v_1}\beta_1,\dots,e^{-v_k}\beta_k)e^{-i(\alpha'_1 v_1+\cdots+\alpha'_k v_k)}\,,
\end{aligned}
\ee
for $F\in\CC_k$, where $\chi_F$ is a smooth function that equals $1$
on a neighbourhood of $K_F$.  
We then have the following relation, whose proof is given in the appendix,
\be\label{antihom}
(m_*F)^* = m_*(F^*)^\wedge\,,\quad F\in\CC_2\,,
\ee
where $^\wedge$ denotes the flip operation on $\CC_2$,
$$
F^\wedge(\alpha_1,\alpha_2,\beta_1,\beta_2) =
F((\alpha_2,\alpha_1,\beta_2,\beta_1)\,.
$$

As a special case we get
\be\label{invalg2}
(f\,*\,g)^* = g^*\,*\,f^*
\ee
for all $f,g\in\CC$.

To summarize, we have established the following result.

\begin{prop}\label{Calg}
 $\CC$ equipped with the $*$-product \eqref{star4} and
$^*$-involution \eqref{geninv} is an involutive algebra.
\end{prop}
   
This algebra can be viewed as an involutive subalgebra of the 
multiplier algebra of $\CB$:

\begin{corollary}
If $f \in \CC$ and $g \in \CB$ then both $f * g$ and $g *f$ are in $\CB$.
\end{corollary}
\begin{proof}
It suffices to show the result only for $f*g$, since both $\CB$ and $\CC$
are involutive algebras. We know that $f*g$ is in $\CC$ so it suffices to 
show that it is a Schwartz function whenever $g$ is. Using (\ref{star4}) 
we have
$$
\begin{aligned}
|f*g(\alpha,\beta)| = &\frac{1}{\sqrt{2\pi}}\left|\int d\alpha'
  f(\alpha',\beta)\CF(\chi_f g_{\alpha,\beta})(\alpha')\right|\\
&\leq C (1+|\beta|)^{M}\Vert \CF(\chi_f g_{\alpha,\beta})\Vert\\
&\leq C'(1+|\beta|)^M \Vert \chi_f g_{\alpha,\beta}\Vert'\,,
\end{aligned}
$$
where $\Vert\cdot\Vert, \Vert\cdot\Vert'$
are appropriate Schwartz norms, $C, C'$ are constants and we have used
\eqref{pol}. If $g$ is a Schwartz function we clearly have 
$$
\Vert \chi_f g_{\alpha,\beta}\Vert' \leq
C_{N',M'}(1+|\alpha|)^{-N'}(1+|\beta|)^{-M'}
$$
for arbitrary $N',M'\geq 0$ and suitable constants $C_{N',M'}$. Hence
$f*g$ is of rapid decrease. Similar arguments apply to derivatives of  \eqref{star4}
w.\,r.\,t. $\alpha,\beta$, thus proving that $f*g$ is a Schwartz
function if $g$ is. 

\end{proof}

\begin{example}\label{Ex1}
It is useful to note the following instances of the $*$-product.                            
\begin{itemize}
\item[a)]~ If $f,g\in \CC$ where $g(\alpha,\beta)=g(\alpha)$ depends only on
  $\alpha$ then $f*g(\alpha,\beta)=f(\alpha,\beta)g(\alpha)$.
\item[b)]~ If $f,g\in\CC$ and $f(\alpha,\beta)=f(\beta)$ depends only on
  $\beta$ then $f*g(\alpha,\beta) =f(\beta)g(\alpha,\beta)$.
\item[c)]~ If $f(\alpha,\beta)=\alpha$ and $g\in\CC$ depends only on
  $\beta$ then
$$ 
(f*g)(\alpha,\beta) = \alpha g(\beta) + i \beta g'(\beta), \;\;\;
(g*f)(\alpha,\beta) = g(\beta) \alpha. 
$$
In particular, the constant function $1$ is a unit of $\CC$ and 
$$ \alpha*g(\beta) - g(\beta)*\alpha  =  i \beta g'(\beta)\,.$$
For $g(\beta)=\beta$ this relation yields
a representation of the defining relation \eqref{kappa_1} in terms of a
$*$-commutator with $t,x$ corresponding to $\alpha,\beta$. Note also that, $\alpha^*= \alpha$
and $\beta^* =\beta$ by \eqref{geninv}. 

As a result we see that the star algebra $\CC$ furnishes a representation 
of $\kappa$-Minkowski space, to be further developed in
Theorem~\ref{iota} below.
\end{itemize}
\end{example}
 
\subsection{The Hopf algebra $\CC$}

We now proceed to discuss the coalgebra structure on $\CC$ using the
same notation for the coproduct, counit and antipode as for $M_\kappa$. 
The coproduct is of the standard cocommutative form 
\be 
(\cop f)(\alpha_1,\alpha_2,\beta_1,\beta_2) = f(\alpha_1+\alpha_2,\beta_1+\beta_2). \label{coproduct} 
\ee
We show below that this defines a map $\cop:\CC \to \CC_2$. The maps 
$\cop\ts 1$ and $1\ts\cop$ have natural extensions to maps from
$\CC_2$ to $\CC_3$, for which we shall use the same notation: 
$$ 
\begin{aligned}
&(\cop\ts 1) f(\alpha_1,\alpha_2,\alpha_3,\beta_1,\beta_2,\beta_3) =
f(\alpha_1+\alpha_2,\alpha_3,\beta_1+\beta_2,\beta_3)\,, \\
&(1\ts\cop)f(\alpha_1,\alpha_2,\alpha_3,\beta_1,\beta_2,\beta_3) =
f(\alpha_1,\alpha_2+\alpha_3,\beta_1,\beta_2+\beta_3)\,.
\end{aligned}
$$

It is evident that $\cop$ is coassociative, that is
$$
(\cop\ts 1) \cop = (1\ts\cop) \cop,
$$
as maps from $\CC_1$ to $\CC_3$.

The counit $\varepsilon: \CC\to\mathbb C$ and antipode $S:\CC\to\CC$ are defined by
\be\label{counit}
\varepsilon(f) = f(0,0)\,,
\ee
and
\be\label{antipode}
(Sf)(\alpha,\beta) = \overline{f^*}(-\alpha,-\beta)\,,
\ee
respectively, for $f\in\CC$. Obviously, $\varepsilon$ and $S$ are
linear maps and it easily verified that 
$$
S^2 = \Id_\CC\,.
$$

We now state the main result of this section.

\begin{theorem}\label{iota} 
$\CC$ equipped with the $*$-product, the $^*$-involution,
the coproduct $\cop$, the counit $\varepsilon$ and the antipode $S$
defined above is a Hopf star algebra. Moreover, the algebra homomorphism 
$\iota$ from $M_\kappa$ to $\CC$ defined by 
$$
\iota(t) = \alpha\,,\qquad \iota(x) = \beta
$$ 
is compatible with the Hopf algebra structure:
\be\label{compat}
 \cop\,\iota = (\iota\ts\iota)\;\cop\,, \qquad \varepsilon\,\iota =
\varepsilon\,,\qquad S\,\iota = \iota\,S\,.
\ee
\end{theorem}

\begin{proof}
To show that $\CC$ is a Hopf algebra we need to check that the
coproduct and counit are well defined, that they are algebra homomorphisms 
and that $S$ satisfies the conditions of an antipode, where the $*$-product on $\CC_2$
is defined by a straightforward generalisation of \eqref{star4} to
functions of $4$ variables, that is
\be
\begin{aligned}\label{cop-2}
&(F*G)(\alpha_1,\alpha_2,\beta_1,\beta_2) = \frac{1}{(2\pi)^2} 
\int dv_1\int dv_2\int d\alpha'_1\int d\alpha'_2\,\\
&~~~~~~~~~~~~~F(\alpha_1 \,+\,\alpha'_1,\alpha_2 \,+\,\alpha'_2,\beta_1,\beta_2)\,G(\alpha_1,\alpha_2,e^{-v_1}\beta_1,e^{-v_2}\beta_2)\,e^{-i(\alpha'_1v_1+\alpha'_2v_2)}\,,
\end{aligned}
\ee
for $F,G\in\CC_2$.

That $\varepsilon$ is a well-defined homomorphism follows by inserting 
$\alpha=\beta=0$ in \eqref{star4}, which gives
$$
\varepsilon(f*g) = \frac{1}{\sqrt{2\pi}} \int dv \, {\tilde f}(v,0) g(0,0) 
                 = f(0,0)\, g(0,0) = \varepsilon(f)\varepsilon(g)\,.
$$
Concerning the coproduct it is evident that $\cop f$ is a smooth
function satisfying the polynomial bounds \eqref{pol}. Noting that
\be
\widetilde{\cop f}(v_1,v_2,\beta_1,\beta_2) = \sqrt{2\pi}\delta(v_1-v_2) \tilde
f(v_1,\beta_1+\beta_2)\,,
\label{copc2}
\ee
where $\delta$ is the Dirac delta function, it follows that
$\widetilde{\cop f}$ has compact $\alpha$-support with 
$$
K_{\cop f} = \{(v,v)\mid v\in K_f\}\,.
$$
Hence $\cop f \in \CC_2$.

To show that $\cop$ is a homomorphism we write \eqref{cop-2} as
$$
\begin{aligned}
&(\triangle f)*(\triangle g)(\alpha_1,\alpha_2,\beta_1,\beta_2) = \\
&= \frac{1}{2\pi} \int  dv_1 dv_2 \, 
{\widetilde{\cop f}}(v_1,v_2,\beta_1,\beta_2)\cop g(\alpha_1,\alpha_2,
e^{v_1}\beta_1, e^{v_2} \beta_2)) e^{i (\alpha_1 v_1 + \alpha_2 v_2)}\,. 
\end{aligned}
$$
Using (\ref{copc2}) this yields
$$
(\cop f)*(\cop g)(\alpha_1,\alpha_2,\beta_1,\beta_2) =
\frac{1}{\sqrt{2\pi}} \int  dv\, \tilde
f(v,\beta_1+\beta_2)g(\alpha_1+\alpha_2, e^{-v}
(\beta_1+\beta_2))e^{i(\alpha_1+\alpha_2)v}\,, 
$$
which is seen to be identical to
$\cop(f*g)(\alpha_1,\alpha_2,\beta_1,\beta_2)$ by \eqref{star4} and
\eqref{coproduct}, as desired.

For the antipode we need to demonstrate the relation
\be
\label{antirel}
m_*\,(S\ts 1)\cop = \varepsilon \Id_\CC = m_*\,(1\ts S)\cop\,,
\ee
where $S\ts 1$ and $1\ts S$ denote the 
natural extensions to $\CC_2$ of the corresponding operators on $\CC\ts\CC$.
First, we use \eqref{geninv} and \eqref{coproduct} to write
$$
(S\ts 1)\cop f(\alpha_1,\alpha_2,\beta_1,\beta_2) =\\
\frac{1}{2\pi}\int d\alpha'\int dv'\chi_f(v')f(\alpha',-e^{v'}\beta_1
+\beta_2)e^{i(\alpha_2-\alpha_1-\alpha')v'}\,.
$$
Since this expression depends on $(\alpha_1,\alpha_2)$ only through $\alpha_1-\alpha_2$
it follows immediately from \eqref{genprod} that $m_*(S\ts 1)\cop f$ is independent of $\alpha$. 
Hence we may set $\alpha=0$ and obtain
\be
\begin{aligned}\label{int0}
&m_*(S\ts 1)\cop f(\alpha,\beta) =\\ &\frac{1}{(2\pi)^2}
\int \! d\alpha'' \! \int \! dv'' \! \int \! d\alpha' \! \int dv' \,
\chi_f(-v'')\chi_f(v')f(\alpha',-e^{v'}\beta \!+\! e^{-v''}\beta)e^{-i\alpha''(v'+v'')}
e^{-i\alpha'v'}.
\end{aligned}
\ee
Let now $\zeta_1$ be a smooth function on $\mathbb R$ of compact
support which equals $1$ on a neighborhood of $0$, and define
the functions $\zeta_R,\, R>0,$ by
$
\zeta_R(v) = \zeta_1(\frac{v}{R})\,, v\in\mathbb R\,.
$
Then insert
$$
1=
(\zeta_{R_1}(\alpha'')+(1-\zeta_{R_1}(\alpha''))(\zeta_{R_2}(\alpha'+\alpha'')+(1-\zeta_{R_2}(\alpha'+\alpha''))
$$

into the integrand and accordingly write the integral as a sum of four
integrals by expanding the product on the right-hand side. By performing an adequate
number of partial integrations w.\,r.\,t. $v', v''$ 
in the three terms containing at least
one factor $(1-\zeta_{R_1}(\alpha''))$ or $(1-\zeta_{R_2}(\alpha'+\alpha''))$,
we obtain absolutely convergent integrals that vanish in
the limit $R_1,R_2\to\infty$. In other words, the integral  \eqref{int0}
 can be obtained as the limit for $R_1, R_2\to\infty$ of the
 absolutely convergent integrals defined by inserting an extra
 convergence factor $\zeta_{R_1}(\alpha'')\zeta_{R_2}(\alpha'+\alpha'')$.
 For the regularized 
integrals we then obtain, after integrating over $\alpha''$, the expression
$$
\frac{1}{(2\pi)^{3/2}}\int dv' dv'' d\alpha'
\chi_f(-v'')\chi_f(v')R_1\CF(\zeta_1)(R_1(v'+v''))\zeta_1(\frac{\alpha'}{R_2})f(\alpha',-e^{v'}\beta +e^{-v''}\beta)e^{-i\alpha'v'}\,. 
$$
It is now easy to verify that the limit for $R_1\to\infty$ equals
$$
\frac{1}{2\pi}\int dv' d\alpha'
\chi_f(v')^2\,\zeta_1(\frac{\alpha'}{R_2})f(\alpha',0)\,e^{-i\alpha'v'}\,.
$$

Finally, letting $R_2\to\infty$ gives the result
$$
\frac{1}{\sqrt{2\pi}}\int d\alpha'\CF((\chi_f)^2)(\alpha')\,f(\alpha',0) = f(0,0)\,,
$$
where we have used that $(\chi_f)^2$ equals $1$ on a neighborhood of $K_f$. This proves the
first equality in \eqref{antirel}. The second one follows similarly. 

So far, we have demonstrated that $\CC$ is a Hopf algebra. To show that it is an
involutive Hopf algebra it remains to verify that $\cop$ and
$\varepsilon$ fulfil
\be
\label{invHopf}
(\cop f)^* = \cop(f^*)\qquad \mbox{and} 
\qquad \varepsilon(f^*) = \overline{\varepsilon(f)}\,,
\ee
where the involution on $\CC_2$ is given by \eqref{geninv-k}. The last relation in
(\ref{invHopf}) is obvious. To establish the former we note that for
$f\in\CC$ and $\tau\in\CS(\mathbb R^4)$ we get from \eqref{coproduct} by a simple change 
of variables that the action of $\widetilde{\cop f}$ as a distribution
on $\tau$ is given by 
$$
\widetilde{\cop f}(\tau) = \int d\alpha d\alpha' d\beta
d\beta'\,f(\alpha,\beta)\tilde{\tau}(\alpha-\alpha',\alpha',\beta-\beta',\beta')\,.
$$
Clearly, $\int d\alpha'd\beta'\tilde{\tau}(\alpha-\alpha',\alpha',\beta-\beta',\beta')$ is a
Schwartz function of $(\alpha,\beta)$ whose Fourier transform w.\,r.\,t. 
$\alpha$ at $(v,\beta)$ equals $\int d\beta'
\tau(-v,-v,\beta-\beta',\beta')$. Thus, setting
$$
(T\tau)(v,\beta) = \int d\beta' \tau(v,v,\beta-\beta',\beta'),
$$
we have
$$
\widetilde{\cop f}(\tau) = \tilde f(T\tau)\,,
$$
and hence by \eqref{inv5}
$$
\widetilde{\cop (f^*)}(\tau) = \widetilde{f^*}(T\tau) = \bar{\tilde f}(R_f T\tau)\,.
$$

Using \eqref{R-op} we have 

\begin{eqnarray*}
\chi_f(v)(R_f T\tau)(v,\beta) &=& \chi_f(v)^2\int d\beta'
\,\tau(-v,-v,e^{-v}\beta-\beta',\beta')\,e^{-v}\\
&=& \chi_f(v)^2\int d\beta'\,\tau(-v,-v,e^{-v}(\beta-\beta'),e^{-v}\beta')\,e^{-2v}\\
&=& T(R_f\ts R_f)\tau (v,\beta)\,.
\end{eqnarray*}

Inserting this into the previous equation yields 
$$
\widetilde{\cop (f^*)}(\tau) = \bar{\tilde f}(T(R_f\ts R_f)\tau) =
\overline{\widetilde{\cop f}}((R_f\ts R_f) \tau) = \widetilde{(\cop f)^*}(\tau)
\,,
$$
which proves the claim. 

Finally, knowing that $\cop, \varepsilon$ are homomorphisms and $S$
an antihomomorphism, the compatibility relations \eqref{compat} follow
by verifying their validity for the generators $t,x$. For $\cop$ and
$\varepsilon$ this is trivial to verify. As noted in Example~\ref{Ex1}
we have $\alpha^* = \alpha$ and $\beta^*=\beta$ such that 
$$
S(\alpha) = -\alpha\qquad\mbox{and}\qquad S(\beta) = -\beta\,.
$$
On the other hand, the antipode on $M_\kappa$ fulfils $St=-t\,,
Sx=-x$, which shows that the relation for $S$ in \eqref{compat} also 
holds for the generators. 
\end{proof}
\begin{remark}
The homomorphism property of $\cop$ proven above holds more generally
in the form
\be
M_*(\cop\ts\cop)F = \cop (m_* F)\,,\quad F\in\CC_2\,,
\ee
where $M_*:\CC_4\to \CC_2$ denotes the canonical extension of
\eqref{cop-2} given by
\be
\begin{aligned}\label{cop-3}
&M_*H(\alpha_1,\alpha_2,\beta_1,\beta_2) = \frac{1}{(2\pi)^2} 
\int \!dv_1 \! \int \! dv_2 \! \int \! d\alpha'_1 \! \int \! d\alpha'_2 \,\\
&~~~~\chi^1_H(v_1)\chi^2_G(v_2)\,H(\alpha_1 \!+\!\alpha'_1,\alpha_2 
\!+\!\alpha'_2,\alpha_1,\alpha_2,\beta_1,\beta_2,e^{-v_1}\beta_1,e^{-v_2}\beta_2)\,
e^{-i(\alpha'_1 v_1+\alpha'_2v_2)}\,,
\end{aligned}
\ee
for $H\in\CC_4$, where $\chi^1_H, \chi^2_H$ denote smooth functions of compact support
that equal $1$ on a neighborhood of the projection of $K_H$ onto the first and second 
axis, respectively.
The verification of \eqref{cop-3} is left to the reader. 
\end{remark}

\begin{remark} We have above used
  the $*$-product and $^*$-involution associated with the right
  invariant Haar measure on $G$ to equip $\CC$ with a Hopf algebra
  structure compatible with that of $\kappa$-Minkowski space. The
  reader may easily check that by similar arguments one obtains
  an alternative realization of $M_\kappa$ on the basis of the
  $\star$-product and $^\star$-involution associated with the left
  invariant Haar measure. 
\end{remark}

\section{Lorentz covariance}\label{sec:4}

A salient  feature of $\kappa$-Minkowski space is the existence of an action 
on it of a deformation of the Poincar\'e Lie algebra \cite{kappafirst,MajRue}, 
the so-called $\kappa$-Poincar{\'e} algebra. In two dimensions, the latter 
is usually presented as the Hopf algebra with generators $E,P$ and $N$, 
the energy, momentum and Lorentz boost, respectively, fulfilling the relations
\be
\begin{aligned}
&[P,E]=0, \;\;\;\;  [N,E] = P, \\
& \cop E = E \ts 1 + 1\ts E\,,\;\; \cop P = P \ts 1 + e^{-E} \ts P\,,\\
&[N,P]=  \frac{1}{2}(1-e^{ -2E}) -\frac{1}{2}P^2,\;\;
 \cop N = N \ts 1 + e^{-E} \ts N\,,
\end{aligned}
\label{poin1}
\ee
and with counit annihilating the generators whereas the antipode acts
according to 

$$ S(E) = -E, \;\;\; S(P) = - e^E P, \;\;\; S(N) = - e^E N. $$

As mentioned previously, the algebraic
$\kappa$-Minkowski space $M_\kappa$ can be defined as the dual of the Hopf
subalgebra generated by $E$ and $P$ \cite{MajRue}. It is the purpose of this section
to exhibit explicitly the action of the $\kappa$-Poincar{\'e} algebra
in terms of linear operators
on the realization $\CC$ of $M_\kappa$ developed in the previous
section. Moreover, we shall find that by
restriction  we obtain an action of the $\kappa$-Poincar{\'e}
algebra on the smaller algebra $\CB$.

To avoid the appearance of the exponential of $E$ in \eqref{poin1}
we prefer to introduce it as an invertible generator $\CE$ and define
the $\kappa$-Poincar{\'e} algebra $\CP_\kappa$ accordingly as the Hopf algebra generated
by $E, P, \CE, N$ fulfilling
\be\label{poin2}
\begin{aligned}
&[P,E] = [P,\CE] = [E, {\cal E}] = 0, \;\;  \\
&[N,E]= P\,,\;\;[N,{\cal E}] =  -{\cal E} P\,, \;\; 
[N, P] = \frac{1}{2}(1-{\cal E}^2) - \frac{1}{2}P^2\,,  \\
& \cop E =E\ts 1 +1\ts E\,, \;\;\;\;
\cop P = P \ts 1 + {\cal E} \ts P\,, \\
& \cop {\cal E} = {\cal E} \ts {\cal E}\,,\;\;\;\;
\cop N = N \ts 1 + {\cal E} \ts N\,,
\end{aligned}
\ee
and with counit and antipode given by
\bq
&\;&\varepsilon(E) = \varepsilon(P)=\varepsilon(N)=0\,,\;\;\;\varepsilon(\CE)
=1\;,\label{co2}\\
&\;&S(E) = -E\,,\;\;\;S(\CE) = \CE^{-1}\,,\;\;\; S(P) = - \CE^{-1} P, \;\;\; S(N) = - \CE^{-1} N\,.\label{anti2}
\eq

We also observe that, although the $\kappa$-Poincar\'e algebra was originally introduced 
without involution, it is easy to verify that 
\be\label{poininv}
E^* = E, \;\; P^* = P, \;\; N^* = -N, \;\;\CE^* =\CE, 
\ee
defines an involution on $\CP_\kappa$ making it a Hopf star
algebra. Note, however, that the involution does not commute with $S$.

\subsection{Action of the momentum subalgebra on $\CC$}

In order to define the action of $P,E,{\cal E}$ on $\CC$ we first
make a slight digression on imaginary translations of elements in
$\CC$.
 
Let $f \in \CC$. Since $\tilde f$ has compact $\alpha$-support it
follows (see e.g. \cite{RS}) that $f$ can be analytically continued to
an entire function of $\alpha$. The analytic continuation will
likewise be denoted  by $f$ and is given by
\be\label{ancont}
f(\alpha +i\gamma, \beta) = \frac{1}{\sqrt{2\pi}}\tilde f(\chi_f(v) e^{i(\alpha
  +i\gamma)v},\beta) =  \frac{1}{\sqrt{2\pi}}\int d\alpha' f(\alpha+\alpha',\beta)
\CF(e^{-\gamma v}\chi_f(v))(\alpha')\,.
\ee

For fixed $\gamma\in\mathbb R$ we claim that the function $T_\gamma f$
defined by
\be\label{Tdef}
(T_\gamma f)(\alpha,\beta) = f(\alpha+i\gamma,\beta)
\ee
belongs to $\CC$. Indeed, since $\CF(e^{-\gamma v}\chi_f(v))(\alpha')$
is a Schwartz function of $\alpha'$ we get immediately from
\eqref{ancont} that the derivatives of $T_\gamma f$ are obtained by
differentiating the integrand and, combining this with \eqref{pol},  it
follows easily that $T_\gamma f$ fulfils polynomial bounds of the
form \eqref{pol}. That $\widetilde{T_\gamma f}$ has compact
$\alpha$-support follows from
\be\label{Tfourier}
\widetilde{T_\gamma f}(v,\beta) = e^{-\gamma v}\tilde f(v,\beta)\,.
\ee
From this relation or, alternatively, from the uniqueness of the
analytic continuation of $f$ in $\alpha$ we conclude that 
the \emph{imaginary translation operators} $T_\gamma:\CC\to\CC$ form a
one-parameter group, 
$$
T_{\gamma+\eta} = T_\gamma T_\eta\,,\quad \gamma,\eta\in\mathbb
R\,,\qquad T_0 = \Id_\CC\,.
$$
Similarly, $n$-parameter groups of imaginary translation operators $T_{\underline{\gamma}}$ are
defined on $\CC_n$ for any $n\in\mathbb N$ and
$\underline{\gamma}\in\mathbb R^n$. We shall write $T_\gamma$  for $T_{(\gamma,\dots,\gamma)}$,
  independently of $n$, for $\gamma\in\mathbb R$.

We next note the following two properties of these maps. 

\begin{prop}
For fixed $\gamma\in\mathbb R$ the map $T_\gamma: \CC\to\CC$ is an
algebra automorphism, that is
\be\label{Tauto}
T_\gamma (m_* F) = m_* (T_\gamma F) \,,\quad F\in\CC_2\,.
\ee
Moreover,
\be\label{Tstar}
 T_\gamma(f^*) = (T_{-\gamma}f)^*\,,\qquad f\in\CC. 
\ee
\end{prop}

\begin{proof}
Since $T_\gamma^{-1} = T_{-\gamma}$, it is sufficient to verify \eqref{Tauto} and \eqref{Tstar}.
By \eqref{genprod} we have
\be\label{RHS2}
 m_* T_\gamma F(\alpha,\beta) = \int d\alpha'\int dv \chi^1_F(v)F(\alpha'
+ \alpha+i\gamma,\alpha +i\gamma,\beta,e^{-v}\beta)\,e^{-i\alpha'v}\,.
\ee
That this is an entire function of $z=\alpha+i\gamma$ for fixed
$\beta$ is seen as
follows. By inserting a convergence factor $\zeta_R(\alpha')$ into the
integrand we have, as seen previously, that the regularized integrals
converge to the integral \eqref{RHS2} as $R\to\infty$ for fixed $z$. It is easy to
see that the convergence is uniform in $z$ on compact subsets of
$\mathbb C$. Since
the regularized integrals are obviously analytic in $z$ it follows that
the same holds for \eqref{RHS2}. Hence this is the unique
entire function whose restriction to $\mathbb R$ coincides with
$m_*F(\alpha,\beta)$ for fixed $\beta$. But this function is by
definition equal to the lefthand side of \eqref{Tauto} for
$z=\alpha+i\gamma$. This concludes the proof of \eqref{Tauto}. 

Concerning \eqref{Tstar} we note that by \eqref{geninv}
\be
(T_{-\gamma} f)^*(\alpha,\beta) = \frac{1}{2\pi}\int d\alpha' \int dv\, \chi_f(-v)
\bar{f}(\alpha-i\gamma +\alpha',e^{-v}\beta)e^{-i\alpha' v},
\label{tgst}
\ee
which by similar arguments as those above is seen to be an entire
function of $z=\alpha +i\gamma$ for fixed $\beta$. Since it coincides with $f^*(\alpha,\beta)$
for $\gamma=0$ we conclude that it equals the lefthand side of 
\eqref{Tstar} for all $z\in \mathbb C$. This proves \eqref{Tstar}.

\end{proof}

By the preceding analyticity argument we obtain 
\be\label{Tfstar}
T_\gamma (f^*)(\alpha,\beta) =  \frac{1}{2\pi}\int d\alpha' \int dv\, \chi_f(-v)
\bar{f}(\alpha', e^{-v}\beta)e^{-\gamma v} e^{i(\alpha-\alpha') v}\,,
\ee
for $f\in\CC$, since the right hand side is seen to be an analytic function of
$z=\alpha+i\beta$ that coincides with the right hand side of
\eqref{tgst} for $\gamma=0$.

\smallskip

Now we can state the main result of this subsection on the action of the
 Hopf subalgebra generated by $E,P,\CE$, called the \emph{extended
momentum algebra}, on $\CC$.

\begin{theorem}\label{Momact}
The algebra $\CC$ is an involutive Hopf module algebra with respect to the
following linear action of the extended momentum algebra on $\CC$:
\be\label{momact}
 E \acts f = -i\frac{\partial f}{\partial \alpha}, \;\;\;
 P \acts f = -i\frac{\partial f}{\partial \beta}, \;\;\;
   {\cal E} \acts f = T_{1}f\,.
\ee
\end{theorem}

\begin{proof}
It is clear that the actions of $E, P, \CE$ defined by \eqref{momact}
are linear on $\CC$ and are mutually commuting. Therefore, it only
remains to verify the compatibility of the action with the
$*$-product and involution. That
$$
E\acts (m_* F) = m_*(E\ts 1)\acts F + m_*(1\ts E)\acts F
$$
is obvious from \eqref{genprod} since differentiation w.r.t. $\alpha$
in the integrand is permitted by a standard convergence argument. 
For the action of $\CE$ we have that 
$$
\CE \acts (m_* F) = m_*(\CE\ts\CE)\acts F,
$$
which is  a special case of \eqref{Tauto}.
Finally, for the action of $P$ we have
\be\label{Pmst}
\begin{aligned}
& \left(P \acts (m_* F)\right)(\alpha,\beta) = \\
& \frac{-i}{(2\pi)^2}\int d\alpha'\int dv \, \chi^1_F(v)
\left(\frac{\partial F}{\partial\beta_1}(\alpha',\alpha,\beta,e^{-v}\beta)+e^{-v}\frac{\partial F}{\partial\beta_2}(\alpha',\alpha,\beta,e^{-v}\beta)\right)e^{i(\alpha-\alpha')v}\,, 
\end{aligned}
\ee
where it is seen that the contribution from the first term in parenthesis 
evidently equals $m_*(P\ts 1)F(\alpha,\beta)$. 

On the other hand, from \eqref{ancont} we get
$$
\left((\CE\ts 1)\acts F\right)(\alpha_1,\alpha_2,\beta_1,\beta_2) = 
\frac{1}{2\pi}\int d\alpha'_1\int dv_1\, \chi^1_F(v_1)
F(\alpha_1+\alpha'_1,\alpha_2,\beta_1,\beta_2)e^{-v_1} e^{-i\alpha'_1v_1} 
$$
and hence 
$$
\begin{aligned}
&m_*\left((\CE\ts 1)\acts F\right)(\alpha,\beta) = 
\frac{1}{2\pi}\int d\alpha'_2\int dv_2\int d\alpha'_1\int dv_1\\
&~~~~~~~~~~~~~~~~~~~~~~~~~~~\chi^1_F(v_1)\chi_F^1(v_2)\,F(\alpha+\alpha'_1+\alpha'_2,\alpha,\beta,e^{-v_2}\beta)e^{-v_1}
e^{-i\alpha'_1v_1+\alpha'_2 v_2}\,.
\end{aligned}
$$
By introducing convergence factors
$\zeta_{R_1}(\alpha'')\zeta_{R_2}(\alpha'+\alpha'')$ as in the proof
of \eqref{antirel} above we obtain after integrating over $\alpha''$
and taking the limit $R_1, R_2\to\infty$ that
$$
m_*\left((\CE\ts 1)\acts F\right)(\alpha,\beta) = \int d\alpha'\int
dv'\chi_F^1(v')^2\,F(\alpha+\alpha',\alpha,\beta,e^{-v'}\beta)e^{-v}e^{-i\alpha'v'}\,.
$$
Using that $(\chi_F^1)^2$ equals $1$ on a neighborhood of the
projection of $K_F$ onto the first axis we see that the second term in
parenthesis in \eqref{Pmst} yields the contribution 
$m_*(\CE\ts P)F(\alpha,\beta)$. Hence, we have shown that
\be\label{Pcompat}
P\acts (m_* F) = m_*(P\ts 1)\acts F + m_*(\CE\ts P)\acts F\,,\quad F\in
\CC_2\,,
\ee
which concludes the argument that the action of the extended momentum
algebra is compatible with multiplication on $\CC$.

Compatibility of the action with the involution is the statement that 
\be\label{compinv}
 (h \triangleright f)^* = (Sh)^* \triangleright f^*\,, 
\ee 
for $f\in \CC$ and $h$ in the extended momentum algebra. It suffices
to verify this for the generators $E,P,\CE$. For $\CE$ it follows directly from \eqref{Tstar}, 
whereas for $E$ it is a consequence of \eqref{geninv} by differentiating both sides with respect to $\alpha$.

Differentiating \eqref{geninv} with respect to $\beta$ and 
using \eqref{Tfstar} we obtain
$$ P \triangleright f^* = - \CE (P \triangleright f)^*\,,$$
which gives
$$ (P \triangleright f)^*= -(\CE^{-1} P) \triangleright f^*\,.$$
Since $S(P)^* = (-\CE^{-1} P)^* = - \CE^{-1} P$ it follows that
\eqref{compinv} is satisfied for $h=P$ . This completes the proof of the theorem.
\end{proof}

\subsection{Action of ${\cal P}_\kappa$ on $\CC$}

To represent the boost operator $N$ by a linear action on $\CC$ we introduce the operators 
of multiplication by $\alpha$ and $\beta$ as

$$ (L_\alpha f)(\alpha,\beta) = \alpha f(\alpha,\beta), \;\;\;\;\;
   (L_\beta f)(\alpha,\beta) =  \beta f(\alpha,\beta),
$$
for $f\in\CC$.

\begin{lemma}\label{ab-rules}
$L_\alpha$ and $L_\beta$ are linear operators on $\CC$ which satisfy the following rules
with respect to the product and involution on $\CC$:
\bq
&\;& L_\alpha (m_*F) = m_*(1\ts L_\alpha) F = m_*(L_\alpha\ts 1) F 
+\,m_*(1\ts L_\beta P) F\,,\label{Laid} \\
&\;& L_\beta (m_* F) = m_*(L_\beta \ts 1) F = m_*({\cal E}^{-1} \ts L_\beta) F\,,
\label{Lbid}\\
&\;& (L_\alpha f)^* =L_\alpha f^* - L_\beta P
f^*\quad\mbox{and}\quad (L_\beta f)^* = \CE L_\beta f^*.
\label{Lstarid}
\eq
\end{lemma}

\begin{proof}
The two left identities in \eqref{Laid} and \eqref{Lbid} follow immediately from \eqref{genprod}. From
the last expression in \eqref{genprod} we obtain 
$$
\begin{aligned}
L_\alpha (m_*F)(\alpha,\beta) &= m_*(L_\alpha\ts 1)F(\alpha,\beta) - \int
d\alpha'\int dv \chi^1_F(v)\alpha'F(\alpha+\alpha',\alpha,\beta,e^{-v}\beta)e^{-i\alpha v}\\
& = m_*(L_\alpha\ts 1)F(\alpha,\beta) + \,m_*(1\ts L_\beta P)F(\alpha,\beta)\,, 
\end{aligned}
$$
where the last step follows by a partial integration
w.\,r.\,t. $v$. This proves the second identity in \eqref{Laid}. Similarly,  the second identity in \eqref{Lbid}
is obtained from
$$
\begin{aligned}
L_\beta (m_*F)(\alpha,\beta) &= \int d\alpha'\int dv
\chi^1_F(v)e^vF(\alpha+\alpha',\alpha,\beta,e^{-v}\beta)e^{-v}\beta
e^{-i\alpha'v}\\& = m_*(\CE^{-1}\ts L_\beta)F(\alpha,\beta)\,,
\end{aligned}
$$
where the last step follows by the same argument as in the proof of 
\eqref{Pcompat} above.

The second identity of \eqref{Lstarid} follows immediately from
\eqref{geninv} and \eqref{Tfstar}. For the first one we multiply both
sides of \eqref{geninv} by $\alpha$ and obtain after a partial
integration 
$$
L_\alpha f^*(\alpha,\beta) = (L_\alpha f)^*(\alpha,\beta) -(L_\beta
P f)^*(\alpha,\beta)\,.
$$
Using the the second identity of \eqref{Lstarid} and \eqref{compinv}
for $h=P$ the first identity of \eqref{Lstarid} follows.
\end{proof}
We are now in a position to extend Theorem~\ref{Momact} as follows.

\begin{theorem}\label{Cmodule}
Defining the linear action of $N$ on $\CC$ by
\be\label{boost}
 N = -i L_\alpha P - \frac{i}{2} (1 - {\cal E}^2) L_\beta +
 \frac{i}{2} L_\beta P^2 \,, 
\ee
and the action of $E,P,\CE$ as in \eqref{momact} then $\CC$
becomes an involutive Hopf module algebra of $\CP_\kappa$.
\end{theorem}

\begin{proof}
That $N, P, E$ and $\CE$ satisfy the commutation relations of
\eqref{poin2} is easily seen by inspection. It remains to check that 
the action of $N$ on $\CC$ is compatible with the product and
involution on $\CC$ using
the coproduct of (\ref{poin2}). By \eqref{Pcompat} and Lemma~\ref{ab-rules} 
one gets, for $F\in\CC$,
$$
\begin{aligned}
N \acts m_* F &= \left( - iL_\alpha P - \frac{i}{2} (1 - {\cal E}^2) L_\beta +
                 \frac{i}{2} L_\beta P^2 \right) m_*F \\
=& -i L_\alpha\, m_*(P\ts 1) F -i L_\alpha m_*({\cal E}\ts P) F \\
 & -\frac{i}{2} L_\beta\, m_*F + \frac{i}{2} {\cal E}^2 m_*(L_\beta \ts 1) F \\
 & + \frac{i}{2}  L_\beta\, m_*\left( (P^2\ts 1) F + 2({\cal E} P \ts P) F 
    + ({\cal E}^2 \ts P^2) F \right)\,. 
\end{aligned}
$$
Making further use of Lemma~\ref{ab-rules} and \eqref{Tauto} this expression equals
$$
\begin{aligned}
& -im_*(L_\alpha P \ts 1) F - i m_*(P \ts L_\beta P) F) 
          -im_*({\cal E} \ts L_\alpha P) F  \\
        &  - \frac{i}{2} m_*(L_\beta\ts 1) F 
           + \frac{i}{2} m_*({\cal E}^2 L_\beta \ts{\cal E}^2) F 
           + \frac{i}{2} m_*(L_\beta P^2 \ts 1) F \\
        &  +i m_*(P \ts L_\beta P) F + \frac{i}{2} m_*({\cal E} \ts L_\beta P^2) F \,.
\end{aligned}
$$ 
Here, two terms are seen to cancel, and using the relation 
$$ m_*({\cal E}^2 L_\beta \ts 1) F - m_*({\cal E} \ts L_\beta) F =0\,,
$$ 
which follows from \eqref{Lbid}, we 
can rewrite the last expression in the form  
$$
\begin{aligned}
& m_*\left( -i(L_\alpha P \ts 1) F - \frac{i}{2} (L_\beta\ts 1) F 
                 +\frac{i}{2} ({\cal E}^2 L_\beta \ts 1) F +
                 \frac{i}{2} (L_\beta P^2\ts 1) F) \right) \\
        & + m_*\left( -i(\CE\ts L_\alpha P) F  -\frac{i}{2} (\CE\ts L_\beta) F
          + \frac{i}{2}(\CE\ts {\cal E}^2 L_\beta) F + \frac{i}{2}
          (\CE\ts L_\beta P^2) F \right))   \\
       &= m_*((N \ts 1)\acts F + ({\cal E} \ts N)\acts F)\,. 
\end{aligned}
$$ 
This proves compatibility of the action of $N$ with the product on $\CC$..

 Using \eqref{compinv} for $h=P$ and $h=\CE$ and \eqref{Lstarid} we get
$$
\begin{aligned}
(N \triangleright f)^* &= \left(\left( -i L_\alpha P - \frac{i}{2} (1 - {\cal E}^2) L_\beta +
 \frac{i}{2} L_\beta P^2 \right) f \right)^* \\
 & = \left( -i L_\alpha \CE^{-1} P + iL_\beta P \CE^{-1} P + \frac{i}{2}(1 - \CE^{-2}) L_\beta \CE
 - \frac{i}{2} L_\beta \CE \CE^{-2} P^2 \right)  f^* \\
 & = \left( -i L_\alpha P - \frac{i}{2}(1-\CE^2)L_\beta
     + \frac{i}{2} L_\beta P^2 \right) \CE^{-1} f^* \\
 &  =  N \CE^{-1} \triangleright f^*,   
\end{aligned}
$$
Noting that $(S(N))^* = - (\CE^{-1} N)^* = N \CE^{-1},$
this proves compatibility of the action of $N$ with involution.
\end{proof}

In view of the obvious fact that $\frac{\partial}{\partial\alpha},
\frac{\partial}{\partial\beta}, L_\alpha, L_\beta$ and $T_1$ all map
$\CB$ into itself, the following is a consequence of Theorem~\ref{Cmodule}.

\begin{corollary}\label{Bmodule}
The subalgebra $\CB$ of $\CC$ is an involutive Hopf module algebra for
$\CP_\kappa$ with action defined by \eqref{momact} and \eqref{boost}.
\end{corollary}

\begin{remark} By inspection of \eqref{poin2}, \eqref{co2},
  \eqref{anti2} it is seen that setting
$$
\Lambda_q(E)=E\,,\;\;\Lambda_q(P)=P\,,\;\;\Lambda_q(\CE)=\CE\,,\;\;\Lambda_q(N)=
N+qP\,,
$$
defines a Hopf algebra automorphism $\Lambda_q$ of $\CP_\kappa$ for each
$q\in\mathbb C$. As a consequence, one obtains an involution on
$\CP_\kappa$ for any $q\in\mathbb R$ by replacing $N^*=-N$ in
\eqref{poininv} by 
\be\label{qinv}
N^*= -N + qP\,.
\ee
For this involution Theorem~\ref{Cmodule} is
still valid if $N$ as given by \eqref{boost} is replaced by
$$N'=N+\frac{q}{2}P\,.$$ 
The  particular choice $q=1$ ensures that the operator $N'$ is
antisymmetric w.\,r.\,t. the $L^2$-inner product on $\CB$, as is easily
verified. More generally, it follows that the action of
$h^*$ on $\CB$ in this case coincides with that of the adjoint of $h$
w.\,r.\,t. the $L^2$-inner product on $\CB$, for any $h\in\CP_\kappa,$
see Proposition~\ref{invintgen} below.
\end{remark}

On $\CB$ the integral w.\,r.\,t. $d\alpha d\beta$ is well defined as a
linear form that we shall denote by $\int$. In the following
proposition we collect some basic properties of $\int$ in relation to
the module algebra structure on $\CB$.

\begin{prop}\label{invintgen}

a)\;\;The integral w.\,r.\,t. the uniform measure on $\mathbb R^2$ is
invariant under the action of $\CP_\kappa$ on $\CB$ defined above in
the sense that, for any $h \in {\cal P}_\kappa$ and $f \in \B$,
\be\label{invtrace}
 \int \, h \acts f = \varepsilon(h) \int \, f\,.
\ee

b)\;\; $\int$ is a left and right invariant integral on the Hopf algebra $\CB$
in the sense that
\be\label{lrinv} 
\left( \int \otimes \Id \right) \cop f =  \int f = 
\left( \Id \otimes \int \right)\cop f\,,\quad f\in\CB\,.
\ee

c)\;\; \be\label{Sinv} \int Sf = \int f\quad\mbox{and}\quad \int f*(Sg) = \int
g*(Sf)\,.\ee

d)\;\; For any $f,g\in \CB$ and $h\in\CC$ we have
\be\label{invadj}\int (h\acts f)*g^* = \int f*(h^*\acts g)^*\,,\ee
if the involution on $\CP_\kappa$ is defined by \eqref{qinv} for $q=1$
and the action of $E, P, \CE, N$ on $\CB$ are given by \eqref{momact} and
\be\label{boost1}
 N\acts f = (-i L_\alpha P - \frac{i}{2} (1 - {\cal E}^2) L_\beta +
 \frac{i}{2} P L_\beta P)f \,,\quad f\in \CB\,. 
\ee

e) \;\; For any $f,g\in \CB$ we have
\be
\label{twitra} \int f*g = \int (\CE \acts g)*f \,.
\ee
which means that $\int$ is a twisted trace.
\end{prop}

\begin{proof}
a)\;\; It suffices to verify \eqref{invtrace} for the generators 
$E,P,\CE$ and $N$. First, since both $E$ and $P$ act on $f$ as 
partial derivatives
\be
\label{hoint}
\int d\alpha d\beta \, (P \acts f)(\alpha,\beta) = 0 = \int d\alpha d\beta \, (E \acts f)(\alpha,\beta)\,.
\ee
For ${\cal E}$ we have
$$ 
\int d\alpha d\beta \, ({\cal E} \acts f)(\alpha,\beta) = 
 \int d\alpha d\beta  f(\alpha +i,\beta) = \int d\alpha d\beta f(\alpha,\beta)
$$
as a consequence of Cauchy's theorem.
Finally, for the action of $N$, one  uses the identities
$$ L_\beta P^2 = P^2 L_\beta - 2 P, \;\;\;\; L_\alpha P = P L_\alpha, $$
to deduce from the preceding results that
$$ \int d\alpha d\beta \, (N \acts f)(\alpha,\beta) = 0. $$
This finishes the proof of a). 

\smallskip

b)\;\; Identities \eqref{lrinv} follow trivially from the translation invariance
of the measure $d\alpha d\beta$.

\smallskip

c)\;\; The first identity of \eqref{Sinv} follows from \eqref{invint}
and \eqref{antipode}:
$$
\int d\alpha d\beta\, (Sf)(\alpha,\beta) = \int d\alpha d\beta\,
(Sf)(-\alpha,-\beta) = \int d\alpha d\beta\,
\overline{f^*}(\alpha,\beta) = \int d\alpha d\beta\, f(\alpha,\beta)\,.
$$
The second identity follows from the former by using that $S$ is an
antihomomorphism and $S^2=\Id$ on $\CB$.

\smallskip

d)\;\; By \eqref{isom} we see that \eqref{invadj} is equivalent to the
statement that the action of $h^*\in\CP_\kappa$ on $\CB$ as a linear
operator on $\CB\subset L^2(\mathbb R^2)$ equals the action of the adjoint of $h$
w.\,r.\,t. the standard inner product on  $L^2(\mathbb R^2)$. That this
holds for $E$ and $P$ is clear form \eqref{momact}. For $\CE$ we have 
$$
\int (\CE\acts f)*g^* = \int \CE\acts (f*(\CE^{-1}\acts g^*) = 
\int f*(\CE \acts g)^*\,,
$$
by \eqref{Tauto}, \eqref{Tstar} and  \eqref{invtrace}. This proves
\eqref{invadj} for $h=\CE$. Since $L_\alpha$ and $L_\beta$ are
symmetric operators on $\CB\subset L^2(\mathbb R^2)$ one can now check
by direct computation that $N$ as given by \eqref{boost} is antisymmetric. 

\smallskip

e) \;\; Using Cauchy's theorem and a change of variables we get from \eqref{star2} that
$$
\begin{aligned}
\int (\CE \acts g)*f 
& = \frac{1}{2\pi} \int d\alpha d\beta \int dv\int d\alpha' \,
    g(\alpha+\alpha'+i,\beta) f(\alpha, e^{-v}\beta) e^{-i\alpha'v}  \\
& = \frac{1}{2\pi} \int d\alpha d\beta \int dv\int d\alpha' \,
    g(\alpha',\beta) f(\alpha, e^{-v}\beta)\,e^{-v} e^{i(\alpha-\alpha')v} 
    \\
& = \int d\beta dv \,
    \tilde g(v,\beta) \tilde f(-v, e^{-v}\beta)\,e^{-v} \,.
\end{aligned}
$$
A change of variables shows that the last expression equals
$
\int d\beta dv \tilde f(v,\beta)\tilde g(-v,e^{-v}\beta)
$
which by reversing the steps above yields $\int f*g$. 
This completes the proof. 
\end{proof} 

\subsection{Explicit dependence on the kappa parameter}\label{parameter}

For the sake of completeness we end this section by reintroducing the $\kappa$-parameter 
 which we eliminated at the outset by rescaling the $t$ generator of
 $M_\kappa$. The correct dependence on $\kappa$ for both
 $M_\kappa$ and $\CP_\kappa$ is obtained by simply rescaling the
 variables $\alpha,\beta$ by $\kappa$, i.\,e. set
 $(\alpha,\beta)=(\kappa\hat\alpha,\kappa\hat\beta)$ and express the
 (co)algebra operations in terms of the dimensionful variables $\hat\alpha,\hat\beta$, and then rename the latter $(\alpha,\beta)$.   
Explicitly, the $*$-product on $\CB$ is replaced by 
 \be
f \ks g(\alpha,\beta) = \frac{1}{2\pi}\int d\alpha' dv \, f(\alpha+\alpha', \beta) 
g(\alpha, e^{-\frac{v}{\kappa}} \beta) e^{-i\alpha v}\,,
\label{kapsta4}
\ee
and the involution is changed to
\be
f^*(\alpha,\beta) = \frac{1}{2\pi}\int d\alpha' dv \, \bar{f}(\alpha+\alpha',
e^{-\frac{v}{\kappa}} \beta) e^{-i\alpha' v}\,,
\ee
whereas the coproduct and counit are unchanged.
Furthermore, the action of the operators $E, P, \CE, N$ on $\CB$ are redefined as
$$
\begin{aligned}
& E \acts f = -i\frac{\partial f}{\partial \alpha}, \;\;\;
  P \acts f = -i\frac{\partial f}{\partial \beta}, \;\;\;
   {\cal E} \acts f = T_{\frac{1}{\kappa}} f\,,\\
& N = -i  L_\alpha P - \frac{i\kappa}{2} \left(1-\CE^2\right) L_\beta +
\frac{i}{2\kappa} L_\beta P^2\,,
\end{aligned}
$$
where $L_\alpha,L_\beta$ denote multiplication by $\alpha,\beta$,
respectively, as before. With these definitions we obtain a function 
algebra realization $\CB$ of $M_\kappa$ and a representation of the 
involutive Hopf algebra $\CP_\kappa$ on $\CB$, as displayed in 
e.\,g. \cite{MajRue}.

\smallskip

Finally, we note the following series representation of the $\ks$-product 
for sufficiently regular functions. For simplicity we consider a rather 
restricted class of functions but the proof can be adapted to more 
general situations. 
 
\begin{prop} If  $f,g\in\CB$ and $g(\alpha,\beta)$ is an entire
  function of $\beta$ then 
$$
 (f \ks g)(\alpha,\beta) = \sum_{n=0}^\infty \frac{i^n}{\kappa^n n!}
\partial_\alpha^n\, f(\alpha,\beta) \left(\beta\partial_\beta\right)^ng (\alpha,\beta), 
$$
for all $(\alpha,\beta)\in\mathbb R^2$.
\end{prop}

\begin{proof}
First rewrite \eqref{kapsta4} as
$$
(f\ks g)(\alpha,\beta) = \frac{1}{\sqrt{2\pi}} \int dv\tilde
f(v,\beta)g(\alpha,e^{-\frac{v}{\kappa}}\beta)\,e^{i\alpha v}\,.
$$
By analyticity of $g(\alpha,e^{-v}\beta)$ in $v$ we have
$$
g(\alpha,e^{-\frac{v}{\kappa}}\beta) = \sum_{n=0}^\infty
\frac{(-1)^n}{\kappa^n n!}v^n (\beta\partial_\beta)^n g(\alpha,\beta)\,.
$$
Inserting this into the previous equation and using that the series is
uniformly convergent on the compact set $K_f$ we get
$$
(f\ks g)(\alpha,\beta) = \frac{1}{\sqrt{2\pi}} \sum_{n=0}^\infty\frac{(-1)^n}{\kappa^n n!} \int dv\tilde
f(v,\beta)v^n (\beta\partial_\beta)^n g(\alpha,\beta)\,e^{i\alpha v}\,.
$$ 
Now, use
$$\frac{1}{\sqrt{2\pi}}\int \tilde f(v,\beta)v^n e^{i\alpha v} =
(-i \partial_\alpha)^n f(\alpha,\beta),$$
to conclude the proof. 
\end{proof}

\section{Conclusions}\label{sec:5}

The star product formulation of the $\kappa$-Minkowski algebra
presented in this paper has potential advantages with regard to future
developments. It is a  basis-independent  
construction realized as a function space with a richer structure than
the algebraic version, and with a
simpler analytic form of the product than in previous approaches. 

We consider it as first step towards the construction of a geometry on $\kappa$-Minkowski
space in the sense of spectral triples. A primary goal will be to study the
equivariant representations of the algebra $\CB$ and to look for equivariant Dirac operators. 
The existence of the invariant twisted trace on $\CB$ suggests that the geometry of 
$\kappa$-Minkowski space might be closer to the case of quantum groups ($q$-deformations)
than originally believed. In particular, the failure of the spectral triple construction
for the compactified version of $\kappa$-Minkowski space is possibly related to 
this fact, and the remedy might be to look  for twisted spectral geometries.

Furthermore, there are interesting relations between the star product formulation of 
$\kappa$-Minkowski space and the deformations of Rieffel
\cite{Rie} determined by actions of $\mathbb R^d$. We postpone the discussion of these
issues, as well as extensions to higher dimensions, to a future publication.

\section{Appendix} \label{App}

\smallskip

The purpose of this appendix is to show that the definitions
\eqref{genprod} and \eqref{geninv} of multiplication and inversion on
$\CC$ are independent of the choice of the functions functions
$\chi_f$ and $\chi_F^1$ satisfying the stated properties and to prove
\eqref{supprod}, \eqref{genass} and \eqref{antihom}.

\bigskip

{\sl\large The support of $m_*F\,.$}

\smallskip

Let $F\in\CC\ts \CC$.
First, observe that by the definition \eqref{genprod} of $m_*$ and the
ensuing convergence arguments we have, for $\varphi\in\CS(\mathbb R^2)$, 
\be\label{A1}
m_*F(\varphi) = \int d\alpha'\int dv'\int d\beta d\alpha
\chi^1_F(v')F(\alpha+\alpha',\alpha,\beta,e^{-v'}\beta)\varphi(\alpha,\beta)e^{-i\alpha'v'}\,.
\ee
For fixed $v',\beta\in\mathbb R$ and $\xi, \eta\in\CS(\mathbb R)$ we have
$$
\int d\alpha' d\alpha F(\alpha+\alpha',\alpha,\beta, e^{-v'}\beta)\CF\xi(\alpha')\CF\eta(\alpha)
= \int du du'\tilde F(u,u',\beta,e^{-v'}\beta)\xi(u)\eta(u+u')\,.
$$
This vanishes if $\eta(u+u')=0$ for all $(u,u')\in K_F$. Since this
holds for arbitrary $\xi\in\CS(\mathbb R)$ it follows that 
$$
\int d\alpha
F(\alpha+\alpha',\alpha,\beta,e^{-v'}\beta)\tilde\varphi(\alpha,\beta) =0\,,
$$
if $\varphi(u+u',\beta)=0$ for all $(u,u')\in K_F$.
Hence we get from \eqref{A1} that $\widetilde{m_*F}(\varphi) =
m_*F(\tilde\varphi)=0$ if $\varphi(u+u',\beta)$ vanishes for
$(u,u')\in K_F$ for arbitrary $\beta$. This proves \eqref{supprod}.  

\bigskip

{\sl\large Independence of the $\chi$-functions.}

\smallskip

Let $f\in\CC$ and write 
$$
f^* = f^*_{1R} + f^*_{2R}\,,
$$ 
where $f^*_{1R}$ and $f^*_{2R}$ are given by \eqref{f1star} and
\eqref{f2star}, respectively, with $\zeta(\alpha')$  replaced by
$\zeta_R(\alpha')=\zeta_1(\frac{\alpha'}{R})$, and where $\zeta_1$ is a
smooth function of compact support that equals $1$ on a neighborhood
of $0$. Choosing $N$ in \eqref{f2star} sufficiently large, it follows
from \eqref{pol} that $f^*_{2R}$ converges to $0$ uniformly on compact
subsets of $\mathbb R^2$ as $R\to\infty$. Hence, $f^*_{1R}$ converges uniformly to
$f^*$ on compact subsets of $\mathbb R^2$. As the reader may easily verify, this also holds if we set 
$\zeta_1 = \CF(\zeta)$, where $\zeta$ is a smooth function with 
support contained in $[-1,1]$ such that $\int_{-\infty}^\infty \zeta(v)dv=1\,,$ since in
this case 
$$
\zeta_R(v)= R\CF(\zeta(Rv))
$$
converges uniformly to $1$ on compact subsets of $\mathbb R$ as
$R\to\infty$. With this choice of $\zeta_R$ we have
$$
f_{1R}^*(\alpha,\beta) =\frac{1}{\sqrt{2\pi}}\int dv \chi_f(-v)\int du\,
R\zeta(R(v-u))\bar{\tilde f}(-u,e^{-v}\beta)\,e^{i\alpha u}\,.
$$
Since the support of $u\to\zeta(Ru)$ is contained in $[-\frac{1}{R},\frac{1}{R}]$
it follows that the last integral vanishes for all $v$ outside 
any given distance $\delta>0$ from $-K_f$ if $R>\frac{1}{\delta}$. 
This proves that the integral defining $f^*$ only depends on the values of
$\chi_f$ in any neighborhood of $K_f$ as desired.

The proof that $m_*F,\,F\in\CC_2,$ only depends on the values of $\chi_F^1$ in any
neighborhood of the projection of $K_F$ onto the first axis is
essentially identical to the preceding argument and we skip further details.

\bigskip

{\sl \large Associativity of the product.}

\smallskip

We consider $m_*$ given by \eqref{genprod} and want to verify the
relation \eqref{genass}.
For $G\in\CC_3$ we have by \eqref{genprod}
$$
(m_*\ts 1)G(\alpha_1,\alpha_2,\beta_1,\beta_2) =
\frac{1}{2\pi}\int d\alpha'_1\int dv_1
\,\chi_G^1(v_1)G(\alpha_1+\alpha_1',\alpha_1,\alpha_2,\beta_1,e^{-v_1}\beta_1,\beta_2)e^{-i\alpha'_1
  v_1}
$$ 
and
\be\label{int1} 
\begin{aligned}
&m_*(m_*\ts 1)G(\alpha,\beta) = \frac{1}{(2\pi)^2}\int d\alpha'_2\int dv_2\int d\alpha'_1\int dv_1\\
&~~~~~~~~~~~\chi_G^1(v_1)\chi_G^{++}(v_2)G(\alpha+
\alpha_1'+\alpha'_2,\alpha+\alpha'_2,\alpha,\beta,e^{-v_1}\beta,e^{-v_2}\beta)e^{-i(\alpha'_1v_1+\alpha'_2
  v_2)}\,,
\end{aligned}  
\ee
where $\chi^{++}_G$ is a smooth function of compact support that
equals $1$ on a neighborhood of the set
$\{v_1+v_2\mid\;(v_1,v_2,v_3)\in K_G\; \mbox{for some $v_3\in\mathbb R$}\}$.
Similarly, we get
\be\label{int2}
\begin{aligned}
&m_*(1 \ts m_*)G(\alpha,\beta) = \frac{1}{(2\pi)^2}\int d\alpha'_1\int dv_1\int d\alpha'_2\int dv_2\\
&~~~~~~~~~~~\chi_G^1(v_1)\chi_G^2(v_2)G(\alpha+\alpha'_1,\alpha+\alpha'_2,\alpha,\beta,e^{-v_1}\beta,e^{-(v_1+v_2)}\beta)e^{-i(\alpha'_1v_1+\alpha'_2 v_2)}
\end{aligned} 
\ee
Now, rewrite \eqref{int1} as
$$ 
\begin{aligned}
&m_*(m_*\ts 1)G(\alpha,\beta) = \frac{1}{(2\pi)^2}\int d\alpha'_2\int
dv_2\int d\alpha'_1\int dv_1\\
&~~~~~~~~~~~\chi_G^1(v_1)\chi_G^{++}(v_2)G(\alpha+
\alpha_1',\alpha+\alpha'_2,\alpha,\beta,e^{-v_1}\beta,e^{-v_2}\beta)e^{-i(\alpha'_1v_1+\alpha'_2
  (v_2-v_1))}
\end{aligned}  
$$
and insert convergence factors $\zeta_R(\alpha_1')\zeta_R(\alpha'_2)$ to justify interchange 
of integrations to obtain
$$ 
\begin{aligned}
&m_*(m_*\ts 1)G(\alpha,\beta) = \frac{1}{(2\pi)^2}\int d\alpha'_1\int dv_1\int d\alpha'_2\int dv_2\\
&~~~~~~~~~~~\chi_G^1(v_1)\chi_G^{++}(v_1+v_2)G(\alpha+
\alpha_1',\alpha+\alpha'_2,\alpha,\beta,e^{-v_1}\beta,e^{-(v_1+v_2)}\beta)e^{-i(\alpha'_1v_1+\alpha'_2
  v_2)}\,.
\end{aligned}  
$$
By an argument similar to the one proving independence of $f^*$ on the
choice of $\chi_f$ above, we may in this integral replace the function
$\chi_G^1(v_1)\chi_G^{++}(v_1+v_2)$ by any smooth function of compact
support that equals $1$ on a neighborhood of the set
$\{(v_1,v_2)\mid\;(v_1,v_2,v_3)\in K_G\; \mbox{for some $v_3\in\mathbb
  R$}\}$. Since this holds for the function
$\chi_G^1(v_1)\chi_G^2(v_2)$ we conclude that the integrals
\eqref{int1} and \eqref{int2} are equal as desired.
\newpage

{\sl\large The $^*$-operation is an antihomomorphisn.}

\smallskip

Let $F\in\CC_2$ and let  $\chi^+_F, \chi_F^1, \chi_F^2$  denote smooth functions of compact support
that equal $1$ on the $\alpha$-support of $m_*F$ and on the
projections of $K_F$ onto the first and second coordinate axis,
respectively.  Using definitions
\eqref{genprod} and \eqref{geninv} we then have 
\be\label{LHS}
\begin{aligned}
(m_*F)^*(\alpha,\beta) = &\frac{1}{(2\pi)^2}\int d\alpha'_2\int dv_2\int
d\alpha'_1\int dv_1\\&\chi^1_F(v_1)\chi^+_F(-v_2)\bar F(\alpha+\alpha'_1,\alpha
+\alpha'_2,e^{-v_2}\beta,e^{-(v_1+v_2)}\beta)
  e^{i\alpha'_1v_1-i\alpha'_2(v_1+v_2)}
\end{aligned}
\ee

For the right-hand side of \eqref{antihom}, on the other hand, we get
\be\label{RHS}
\begin{aligned}
&m_*((F^*)^\wedge)(\alpha,\beta) = \\&\frac{1}{(2\pi)^3}\int d\alpha'\int dv\int
d\alpha'_2\int dv_2\int d\alpha'_1\int dv_1\,\chi^2_F(-v)\chi^2_F(-v_2)\chi^1_F(-v_1)\\
&\qquad~~~~~~~\bar F(\alpha
+\alpha'_1,\alpha'_2,e^{-(v_1+v)}\beta,e^{-v_2}\beta)e^{-i(\alpha'_1v_1+\alpha'_2v_2)+i\alpha
    v_2}e^{i\alpha'(v_2-v)}\,.
\end{aligned}
\ee
Inserting convergence factors
$\zeta_{R_1}(\alpha'_1)\zeta_{R_2}(\alpha'_2)\zeta_{R}(\alpha')$ into
the last integral we recover its value in the limit $R,R_1,R_2\to
\infty$ by the same arguments as above. By performing the $\alpha'$-integration first in the
regularized integral we obtain
$$
\begin{aligned}
&~~~ \frac{1}{(2\pi)^{5/2}}\int dv\int
d\alpha'_2\int dv_2\int d\alpha'_1\int dv_1\chi^2_F(-v)\chi^2_F(-v_2)\chi^1_F(-v_1)\,\\
&\zeta_{R_2}(\alpha'_2)\zeta_{R_1}(\alpha'_1)\bar F(\alpha
+\alpha'_1,\alpha'_2,e^{-(v_1+v)}\beta,e^{-v}\beta)e^{-i(\alpha'_1v_1+\alpha'_2v_2)+i\alpha
    v}\CF(\zeta_{R})(v-v_2)\,,
\end{aligned}
$$ 
and in the limit $R\to\infty$ this gives
$$
\begin{aligned}
&~~~ \frac{1}{(2\pi)^{2}}\int
d\alpha'_2\int dv_2\int d\alpha'_1\int dv_1\,\chi^2_F(-v_2)^2\chi^1_F(-v_1)\\
&\zeta_{R_2}(\alpha'_2)\zeta_{R_1}(\alpha'_1)\bar F(\alpha
+\alpha'_1,\alpha'_2,e^{-(v_1+v_2)}\beta,e^{-v_2}\beta)e^{-i(\alpha'_1v_1+\alpha'_2v_2)+i\alpha
    v_2}\,.
\end{aligned}
$$
A simple change of variables now yields 
$$
\begin{aligned}
&m_*((F^*)^\wedge)(\alpha,\beta) = \lim_{R_1,R_2\to\infty}\frac{1}{(2\pi)^{2}}\int
d\alpha'_2\int dv_2\int d\alpha'_1\int dv_1\,\chi^2_F(-v_1-v_2)^2\chi^1_F(v_1)\\
&\zeta_{R_2}(\alpha'_2)\zeta_{R_1}(\alpha'_1)\bar F(\alpha
+\alpha'_1,\alpha+\alpha'_2,e^{-v_2}\beta,e^{-(v_1+v_2)}\beta)e^{i\alpha'_1v_1-i\alpha'_2(v_1+v_2)}\,,
\end{aligned}
$$
Repeating previous arguments we see by choosing $\zeta_R$ such that $\CF(\zeta_R)$ has
support in $[\frac 1R, \frac 1R]$ that in the limit above the function
$\chi_F^2(-v_1-v_2)^2\chi_F^1(v_1)$ can be replaced by any smooth function
of compact support that equals $1$ on a neighborhood of the set $\{(v_1,v_2)\mid (v_1,-v_1-v_2)\in
K_F\}$ without changing the value of the limit. Since this holds, in particular, for the function
$\chi^1_F(v_1)\chi^+_F(-v_2)$ we conclude that the limit equals
\eqref{LHS}. This proves \eqref{antihom}.
 
\bigskip

{\bf Acknowledgement} ~  This work was supported in part by a Marie Curie Transfer 
of Knowledge project MTKD-CT-42360 and the Polish Government grant 1261/7.PRUE/2009/7.


\end{document}